\documentclass[showpacs,aps,graphicx,twocolumn]{revtex4-2}
\usepackage[tbtags]{amsmath}
\usepackage{mathrsfs}
\usepackage{amsfonts}
\usepackage{amssymb}
\usepackage{amsthm}
\usepackage{subfigure}
\usepackage{eufrak}
\usepackage{multirow}
\usepackage{float}
\usepackage[percent]{overpic}
\usepackage{bm}
\usepackage{xcolor}
\usepackage{blkarray}
\usepackage[colorlinks,linkcolor=blue,anchorcolor=blue,citecolor=blue,urlcolor=blue]{hyperref}
\usepackage{soul,color,xcolor}
\usepackage{booktabs}
\theoremstyle{remark}
\newtheorem{theorem}{Theorem}[section]   
\newtheorem{lemma}[theorem]{Lemma}       
\newtheorem{corollary}[theorem]{Corollary}             
\newtheorem*{remark}{Remark}   

\begin{document}

\title{Efficient High-Dimensional Quantum Circuit Synthesis: From Multi-Controlled Gates to Isometries and Quantum Channels}

\author{Gui-Long Jiang\textsuperscript{1,2}}
\email[]{jianglongabs@gmail.com}
\address{\textsuperscript{\rm1} Institute for Advanced Study in Mathematics, Harbin Institute of Technology, Harbin 150001, China\\
\textsuperscript{\rm2} School of Mathematics, Harbin Institute of Technology, Harbin 150001, China}

\begin{abstract}
Circuit synthesis of multi-controlled gates is crucial for qudit ($d$-level) quantum computing.
This paper presents efficient synthesis schemes that reduce the elementary gate count for multi-controlled single-qudit gates.
For synthesizing general $(n-1)$-controlled unitaries on $n$ qudits, we reduce the controlled-increment (CINC) and generalized controlled-$X$ (GCX) gate counts to $O(n^2)$, improving upon existing $O(n^{2+\log_2 d})$ CINC and $O(n^3)$ GCX bounds.
For $(n-1)$-controlled special unitaries, this complexity is further reduced to $O(n)$.
Furthermore, we present a method for approximately synthesizing multi-controlled unitary gates with a linear complexity in the number of control qudits.
By utilizing the proposed circuit, we present qudit-based circuit constructions for isometries and quantum channels from $n$ to $m$ qudits.
Moreover, for the first time, we present a circuit synthesis scheme for single-controlled gates using SUM gates and single-qudit gates when $d$ is prime.
This enables all CINC-based circuits for various quantum operations to be converted into SUM-gate circuits while preserving the same asymptotic complexity.
Finally, we establish a theoretical lower bound on the number of SUM and CINC gates required to synthesize general $n$-qudit unitaries.
\end{abstract}

\maketitle

\section{Introduction}\label{sec1}

The quantum circuit model serves as the foundational framework for executing and analyzing quantum computation \cite{Nielsen2000}.
As quantum algorithms and hardware inherently rely on performing sequential unitary operations on quantum states, the circuit synthesis of general unitary gates is central to characterizing computational complexity and optimizing hardware design \cite{Nielsen2000,Dawson2005}.
Typically, a universal set of elementary gates is selected, enabling any target unitary operation to be implemented or approximated within a desired precision by a quantum circuit consisting only of these gates \cite{Barenco1995,Brylinski2001}.
In this context, minimizing the count of elementary gates required for quantum operations is a key objective, as it directly reduces the algorithmic overhead and minimizes noise-induced errors in physical quantum systems.

For the qubit case, the controlled-not (CNOT) gate and all single-qubit gates are usually taken as the elementary gates \cite{Barenco1995}, and some quantum circuits for general $n$-qubit unitary operations have been constructed using these gates \cite{Khaneja2001,Vartiainen2004,Mottonen2004,Shende2006,Cartan2023,Sun2023,Krol2024}.
Multi-controlled unitary gates often appear as intermediate building blocks in the decomposition of general unitary operations \cite{Vartiainen2004,Mottonen2004,Shende2006,Chen2015,Bullock2005,Li-Wen-Dong} and have been widely used in quantum algorithms \cite{Chuang2021,high-algorithm1,high-algorithm2,Rosa2025}.
Much work has been done on reducing the number of CNOT gates required for quantum circuits of multi-controlled single-qubit gates \cite{Barenco1995,isometries,Vale2024,Silva2025} and on optimizing the circuit depth \cite{Saeedi2013,Silva2022,Claudon2024}.
For an $(n-1)$-controlled unitary gate, the best-known upper bound on the number of CNOTs required for its synthesis is $4n^2-12n+10$ \cite{Silva2022}.
When the target gate is a special unitary, the best-known CNOT upper bound is $16n-40$ \cite{Rosa2025}.

Compared with the qubit case, the corresponding problem for qudits ($d$-level quantum systems, where $d\geq2$) has received less attention.
Several generalizations of the CNOT gate to the qudit case have been proposed, including SUM gates \cite{SUM-gate} (also known as generalized XOR gates \cite{GXOR-gate}), controlled-increment (CINC) gates \cite{Brennen2005}, and generalized controlled-$X$ (GCX) gates \cite{Di2013}.
In 2006, Brennen \emph{et al.} \cite{Brennen2006} generalized the method of Ref.~\cite{Barenco1995} to qudits and proposed a quantum circuit for $(n-1)$-controlled single-qudit gates using $O(n^{2+\log_{2}d})$ CINC gates.
Di \emph{et al.} \cite{Di2013} proved an upper bound of $O(n^3)$ GCX gates for synthesizing an $(n-1)$-controlled single-qudit gate.
Recently, it has been shown that for odd $d$, an $(n-1)$-controlled INC gate can be synthesized using only $O(n)$ GCX gates (see \cite[Theorem III.4]{Multi-qudit2023}), while for even $d$, the $O(n)$ bound also holds but with the assistance of one ancillary qubit (see \cite[Theorem III.1]{Multi-qudit2023}).
In contrast to the qubit case, none of the above qudit results using CINC or GCX gates achieve the upper bound of $O(n^2)$.
Moreover, a quantum circuit for multi-controlled single-qudit gates using only SUM gates and single-qudit gates has not yet been proposed.

\emph{Contributions.---}In this paper, we address both of the above open issues.
For any $d \geq 2$, we first propose a circuit construction for single-controlled unitary gates using two CINC gates and four single-qudit gates, which generalizes the corresponding qubit result (see \cite[Lemma 5.1]{Barenco1995}) to qudits.
Furthermore, we construct explicit quantum circuits for $(n-1)$-controlled special unitaries that achieve an $O(n)$ CINC gate complexity.
For general $(n-1)$-controlled unitaries, we present a recursive circuit synthesis with at most $O(n^2)$ CINC gates.
By truncating the recursive decomposition, we obtain an approximate synthesis method for $(n-1)$-controlled unitary gates, where the CINC gate complexity scales linearly with the number of control qudits.
Compared with the bound established by Brennen \emph{et al.} \cite{Brennen2006}, our construction reduces the CINC count for synthesizing an $(n-1)$-controlled unitary gate from $O(n^{2+\log_{2}d})$ to $O(n^2)$, with the exact gate count provided in a closed-form expression.
Since a CINC gate can be decomposed into $d-1$ GCX gates \cite{Di2013}, our $O(n^2)$ bound on the CINC gate count consequently reduces the GCX count from $O(n^3)$ \cite{Di2013} to $O(n^2)$ for synthesizing any $(n-1)$-controlled unitary gate.
In particular, for qubits (i.e., $d=2$), our synthesis of an $(n-1)$-controlled special unitary gate requires only $16n-48$ CNOT gates.
As applications, we extend the qubit-based circuit constructions for isometries \cite{isometries} and quantum channels \cite{channels} to qudits.
Table~\ref{Table1} summarizes our main results on the number of CINC gates required for synthesizing quantum operations. 
In particular, for general unitaries acting on $n$ qudits, our CINC gate count is fewer than that reported in \cite{Brennen2006}.

Moreover, when $d$ is prime, we present a synthesis scheme for single-controlled unitary gates using $d$ SUM gates and $d+2$ single-qudit gates.
Consequently, in this case, by decomposing each CINC gate into $d$ SUM gates, the corresponding SUM gate counts for synthesizing the quantum operations listed in Table~\ref{Table1} can be directly obtained by multiplying the CINC gate counts by $d$, preserving the same asymptotic complexity.
We also establish a theoretical lower bound on the number of SUM gates required for synthesizing general unitaries acting on $n$ qudits.
Specifically, any quantum circuit composed of SUM gates and single-qudit gates (with free parameters) that can implement any $n$-qudit gates must contain at least $\lceil\frac{1}{2d(d-1)}(d^{2n}-nd^2+n-1)\rceil$ SUM gates.
When $d=2$, this bound coincides with the qubit lower bound given in \cite[Proposition 1]{Shende2004}.
The same lower bound also holds for CINC-based circuits.

This paper is organized as follows. 
Section~\ref{sec2} provides preliminaries and introduces the qudit gates used in this work.
In Sec.~\ref{sec3}, we propose explicit quantum circuits for single-controlled unitary gates.
Section~\ref{sec4} presents synthesis schemes for multi-controlled unitary and special unitary gates.
Section~\ref{sec5} constructs quantum circuits for isometries and quantum channels; it also establishes a theoretical lower bound on the number of SUM and CINC gates.
Finally, Section~\ref{sec6} concludes the paper.

\begin{table*}[htbp]
\centering
\caption{Asymptotic upper bounds on the number of CINC gates required for synthesizing various quantum operations on $n$ qudits. Here, $K$ denotes the Choi rank of a quantum channel, and $\lceil \cdot \rceil$ represents the ceiling function. The precise gate counts are given in the corresponding theorems and corollaries in the last column.}\label{Table1}
\begin{tabular}{l@{\hspace{5\tabcolsep}}c@{\hspace{5\tabcolsep}}c@{\hspace{5\tabcolsep}}l}
\toprule
Operation & Ancilla & CINC count & Reference \\ 
\midrule
$(n-1)$-controlled special unitary & 0 & $O(n)$ & Theorem~\ref{The4.6} \\
$(n-1)$-controlled unitary & 0 & $O(n^2)$ & Theorem~\ref{The4.7} \\
$(n-1)$-controlled unitary & 1 & $O(n)$ & Corollary~\ref{Cor4.9} \\
Approximate $(n-1)$-controlled unitary & 0 & $O(n)$ & Theorem~\ref{The4.10} \\
$n$-qudit state preparation   & 0 & $O(d^n)$ & Theorem~\ref{The5.1} \\
Isometry from $n$ to $m$ qudits ($n \leq m$) & $m-n$ & $O(d^{n+m})$ & Theorem~\ref{The5.3} \\
Quantum channels from $n$ to $m$ qudits & $\lceil \log_d K \rceil + m - n$ & $O(d^{n+m+\lceil \log_d K \rceil})$ & Theorem~\ref{The5.6} \\
\bottomrule
\end{tabular}
\end{table*}

\section{Preliminaries}\label{sec2}

\subsection{Notations}\label{sec2.1}

A qudit is a $d$-level quantum system and its associated quantum state space is a $d$-dimensional Hilbert space denoted by $\mathcal{H}_{d}$.
An $n$-qudit is a composite quantum system made up of $n$ qudits, with the associated state space being a $d^{n}$-dimensional Hilbert space $\mathcal{H}_{d}^{\otimes n}$. 
Let $[\underline{d}]$ denote the set $\{0, 1, \ldots, d-1\}$, and let $\{|x_{1} \ldots x_{n}\rangle : x_{1} \ldots x_{n} \in [\underline{d}]^{n}\}$ be the canonical basis of $\mathcal{H}_{d}^{\otimes n}$, also called the computational basis of the $n$-qudit. 
An $n$-qudit pure state is a unit vector $|\psi\rangle \in \mathcal{H}_{d}^{\otimes n}$ satisfying $\||\psi\rangle\|^2=\langle\psi|\psi\rangle=1$.
We denote the set of all $n$-qudit pure states by $\mathcal{S}(\mathcal{H}_{d}^{\otimes n})$.
More generally, when a quantum system is in a statistical ensemble of pure states, it is described as a mixed state and represented by a density operator $\rho$ acting on $\mathcal{H}_{d}^{\otimes n}$.
The set of all such density operators on $\mathcal{H}_{d}^{\otimes n}$ is denoted by $\mathcal{D}(\mathcal{H}_{d}^{\otimes n})$.

An $n$-qudit gate is a unitary operator acting on $\mathcal{H}_{d}^{\otimes n}$, and we denote the set of all unitaries by $\mathrm{U}(\mathcal{H}_{d}^{\otimes n})$.
Let $I_{d^{n}}$ denote the identity operator acting on $\mathcal{H}_{d}^{\otimes n}$.
An isometry $V$ from an $n$-qudit to an $m$-qudit system ($n\leq m$) is a linear operator from $\mathcal{H}_{d}^{\otimes n}$ to $\mathcal{H}_{d}^{\otimes m}$ satisfying $V^{\dag}V=I_{d^n}$.
The set of all such isometries is denoted by $\mathrm{U}(\mathcal{H}_{d}^{\otimes n},\mathcal{H}_{d}^{\otimes m})$.
When $n=m$, the isometry $V$ reduces to an $n$-qudit gate.
A quantum channel $\mathcal{N}$ from $n$ to $m$ qudits, where $n$ and $m$ are arbitrary positive integers, is a linear, completely positive, and trace-preserving map from $\mathcal{D}(\mathcal{H}_{d}^{\otimes n})$ to $\mathcal{D}(\mathcal{H}_{d}^{\otimes m})$.
The Choi rank of $\mathcal{N}$ is defined as $\mathrm{rank}(J_{\mathcal{N}})$, where $J_{\mathcal{N}}$ is the Choi matrix of $\mathcal{N}$ given by $J_{\mathcal{N}}=\sum_{x,y\in [\underline{d}]^{n}} \mathcal{N}(|x\rangle\langle y |) \otimes |x\rangle\langle y |$.

\subsection{Qudit Gates and Circuits}\label{sec2.2}

Next, we introduce several useful single-qudit gates.
The \emph{increment} (INC) gate $X_{d}\in \mathrm{U}(\mathcal{H}_{d})$ and the quantum Fourier transform $F_{d}\in \mathrm{U}(\mathcal{H}_{d})$ are defined as \cite{Gottesman1999}
\begin{equation}\label{eq1}
\begin{split}
&X_{d}=|0\rangle\langle d-1| +\sum^{d-2}_{a=0}|a+1\rangle\langle a|,\\
&F_{d}=\frac{1}{\sqrt{d}}\sum_{a,b\in [\underline{d}]}\exp(ab\frac{2\pi \mathrm{i}}{d})|a\rangle\langle b|.
\end{split}
\end{equation}
The increment gate $X_{d}$ can be diagonalized with the action of conjugation by $F_{d}$ \cite{Gottesman1999}, i.e.,
\begin{equation}\label{eq2}
F_{d} X_{d} F^{\dag}_{d}=\sum_{a\in [\underline{d}]}\exp(a\frac{2\pi \mathrm{i}}{d})|a\rangle\langle a|=:Z_{d}.
\end{equation}
Hereafter, the symbol $\dagger$ denotes the conjugate transpose.
Moreover, we define a unitary and Hermitian operator $T_{d}=T_{d}^{\dag}=|0\rangle\langle 0| +\sum^{d-1}_{a=1}|a\rangle\langle d-a|$ on $\mathcal{H}_{d}$.
By the definitions of $X_{d}$ and $T_{d}$, it is observed that
\begin{equation}\label{eq3}
X^{d-1}_{d}=X^{\dag}_{d}=T_{d} X_{d} T_{d}.
\end{equation}

Given $a\in [\underline{d}]$, two distinct indices $1\leq i,j \leq n$, and a single-qudit gate $U \in \mathrm{U}(\mathcal{H}_{d})$, we define a single-controlled qudit gate $\mathrm{C}^{n,a}_{i,j}(U) \in \mathrm{U}(\mathcal{H}_{d}^{\otimes n})$ by its action on the computational basis: for every $x_{1} \dots x_{n} \in [\underline{d}]^{n}$,
\begin{equation}\label{eq4}
\mathrm{C}^{n,a}_{i,j}(U)|x_{1} \dots x_{n}\rangle =
\begin{cases}
|x_{1} \dots x_{n}\rangle, & \text{if } x_i \neq a, \\[1.5mm]
(U)_j|x_{1} \dots x_{n}\rangle, & \text{if } x_i = a,
\end{cases}
\end{equation}
where $(U)_j$ indicates that the operator $U$ acts non-trivially only on the $j$-th qudit.
Here the subscripts $i$ and $j$ represent the positions of the control qudit and the target qudit, respectively.
When $a=d-1$, we omit the superscript $a$ and write $\mathrm{C}^{n,d-1}_{i,j}(U)=\mathrm{C}^{n}_{i,j}(U)$.
In fact, $\mathrm{C}^{n,a}_{i,j}(U)$ can be transformed into $\mathrm{C}^{n}_{i,j}(U)$ by a pair of single-qudit gates acting on the control qudit as follows: 
\begin{equation}\label{eq5}
\mathrm{C}^{n,a}_{i,j}(U) = (X_{d}^{a+1})_{i} \cdot \mathrm{C}^{n}_{i,j}(U) \cdot (X_{d}^{d-1-a})_{i}.
\end{equation}
Fig.~\ref{1a} illustrates the quantum circuit corresponding to Eq.~\eqref{eq5} for the case of $(n,i,j)=(2,1,2)$.
Generally, we can similarly define a multi-controlled qudit gate by taking the indices $a$ and $i$ as tuples instead of single integers.
For example, Fig.~\ref{1b} illustrates the quantum circuit symbols of $\mathrm{C}^{4}_{(1,2,3),4}(U)$, $\mathrm{C}^{4}_{(1,2,4),3}(U)$, and $\mathrm{C}^{4,(a,b)}_{(1,2),4}(U)$ from left to right, respectively.

\begin{figure} [htbp]
  \centering
  \subfigure[]{
  \includegraphics[width=4.8cm]{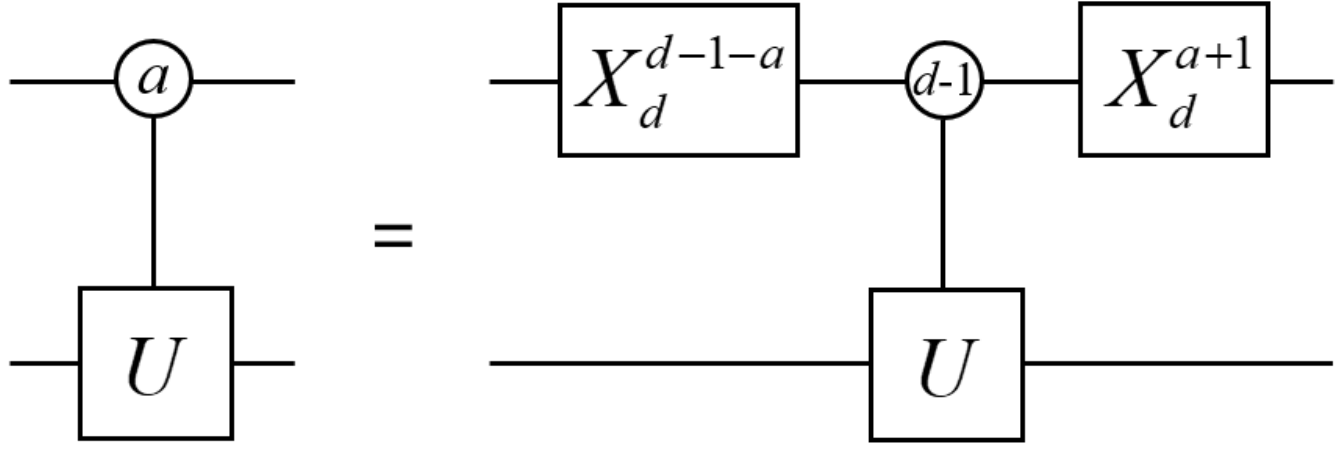}
  \label{1a}}
  \subfigure[]{
  \includegraphics[width=3.2cm]{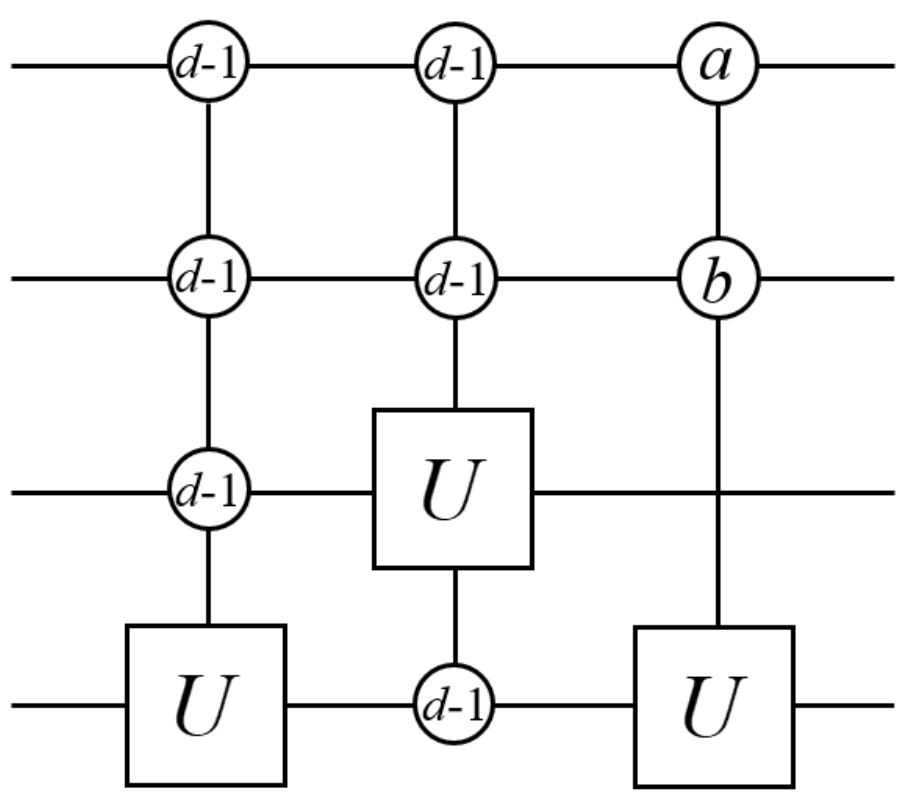}
  \label{1b}}
  \caption{(a) Circuit representation of Eq.~\eqref{eq5} for the case $(n,i,j)=(2,1,2)$. The circle represents the control qudit. (b) Quantum circuit symbols of $\mathrm{C}^{4}_{(1,2,3),4}(U)$, $\mathrm{C}^{4}_{(1,2,4),3}(U)$, and $\mathrm{C}^{4,(a,b)}_{(1,2),4}(U)$.}
  \label{Fig.1}
\end{figure}

In particular, $\mathrm{C}^{n}_{i,j}(X_{d})$ is known as the controlled-increment (CINC) gate.
By Eq.~\eqref{eq3}, $\mathrm{C}^{n}_{i,j}(X_{d})$ is equivalent to $\mathrm{C}^{n}_{i,j}(X_{d}^{\dagger})$ up to conjugation by $T_{d}$ on the target qudit, as shown in Fig.~\ref{2a} for $(n,i,j)=(2,1,2)$.
Moreover, a SUM gate is defined as \cite{Gottesman1999}
\begin{equation}\label{eq6}
\mathrm{SUM}
=\sum_{a\in [\underline{d}]} |a\rangle\langle a| \otimes X_{d}^{a}.
\end{equation}
We use the quantum circuit symbol shown in Fig.~\ref{2b} to represent a SUM gate.
For the qubit case (i.e., $d=2$), the CINC and SUM gates both reduce to the CNOT gate.

\begin{figure} [htbp]
  \centering  
  \subfigure[]{
  \includegraphics[width=4.5cm]{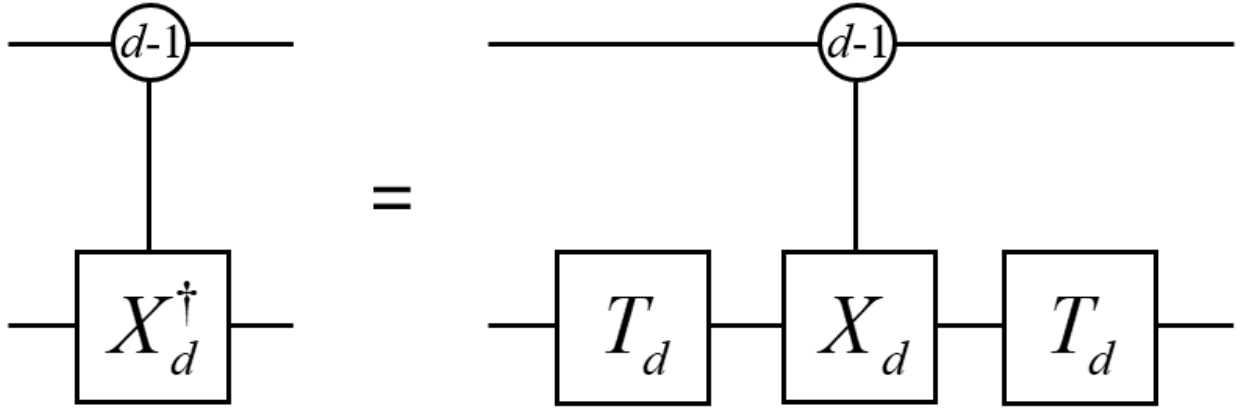}
  \label{2a}}\quad
  \subfigure[]{
  \includegraphics[width=1.3cm]{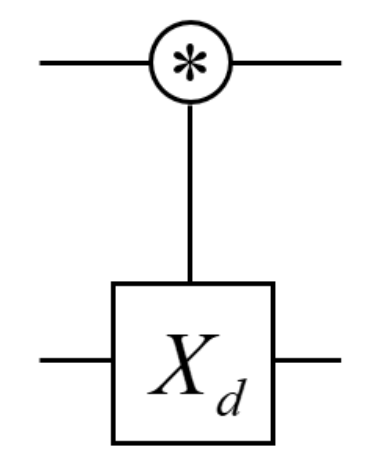}
  \label{2b}}
  \caption{(a) Relation between $\mathrm{C}^{2}_{1,2}(X_{d})$ and $\mathrm{C}^{2}_{1,2}(X_{d}^{\dagger})$. (b) Quantum circuit symbols of the SUM gate.}
  \label{Fig.2}
\end{figure}

\section{Quantum Circuits for Single-controlled Qudit Gates}\label{sec3}

\subsection{CINC-based Quantum Circuits for Single-Controlled Unitaries}\label{sec3.1}

In this subsection, a synthesis of the single-controlled qudit gate $\mathrm{C}^{n}_{i,j}(U)$ is presented by using two CINC gates and four single-qudit gates.
Without loss of generality, we consider the case of $(n,i,j)=(2,1,2)$.
A similar decomposition was first proposed in \cite{Wei2026}.
However, the method presented here is simpler and, when $U \in \mathrm{SU}(\mathcal{H}_{d})$, the single-qudit gate acting on the control qudit can be removed.

\begin{lemma}\label{Lem3.1}
For every $U\in \mathrm{U}(\mathcal{H}_{d})$, the single-controlled qudit gate $\mathrm{C}^{2}_{1,2}(U)$ can be decomposed into
\begin{equation}\label{eq7}
\begin{split}
\mathrm{C}^{2}_{1,2}(U)=&\;(C \otimes AB)\cdot\mathrm{C}^{2}_{1,2}(X_d) \\
&\cdot I_d \otimes B^{\dag} \cdot \mathrm{C}^{2}_{1,2}(X_{d}^{\dag}) \cdot I_d \otimes A^{\dag},
\end{split}
\end{equation}
where $A,B\in \mathrm{SU}(\mathcal{H}_{d})$, $C\in \mathrm{U}(\mathcal{H}_{d})$, and $B,C$ are diagonal in the computational basis.
The quantum circuit representation is shown in Fig.~\ref{Fig.3}.
Moreover, the single-qudit gate $C$ can be omitted if and only if $U\in \mathrm{SU}(\mathcal{H}_{d})$.
\end{lemma}

\begin{figure} [htbp]
  \centering
  \includegraphics[width=6.6cm]{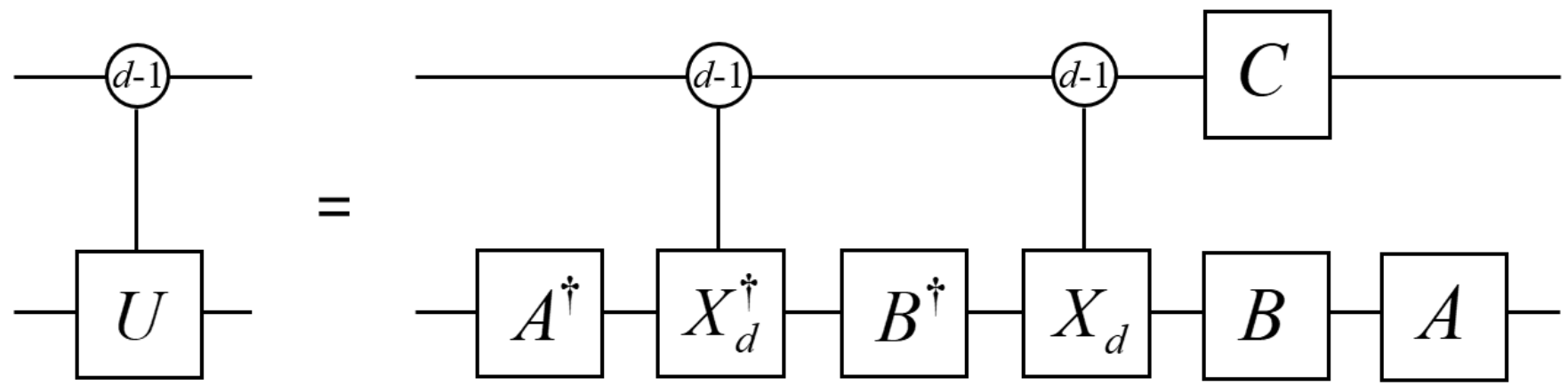}  
  \caption{Equivalent quantum circuit for the controlled $U$ gate based on Eq.~\eqref{eq7}.}
  \label{Fig.3}
\end{figure}

\begin{proof}
Let $U=ADA^{\dag}$ where $A$ is a unitary and $D=\sum_{a=0}^{d-1}e^{\mathrm{i}\theta_{a}}|a\rangle\langle a|$. 
If $\det(A)=e^{\mathrm{i}\theta}\neq 1$, one can replace $A$ by $e^{-\mathrm{i}\frac{\theta}{d}}A \in \mathrm{SU}(\mathcal{H}_{d})$.

Assume first that $U\in \mathrm{SU}(\mathcal{H}_{d})$, and it follows that $D\in\mathrm{SU}(\mathcal{H}_{d})$.
In this case, $\sum_{a=0}^{d-1}\theta_{a}=2k\pi$ for some integer $k$.
Thus $D$ can be expressed as
\begin{equation}\label{eq8}
D=e^{-\mathrm{i}\sum_{a=1}^{d-1}\theta_{a}}|0\rangle\langle 0|+\sum_{a=1}^{d-1}e^{\mathrm{i}\theta_{a}}|a\rangle\langle a|.
\end{equation}
Define a diagonal operator $B=\sum_{i=0}^{d-1}e^{\mathrm{i}\alpha_{i}}|i\rangle\langle i|$ where
\begin{equation}\label{eq9}
\begin{cases}
\alpha_{i}=\frac{1}{d}\sum_{a=1}^{d-1}a\theta_{a}-\sum_{b=i+1}^{d-1}\theta_{b},  & \text{if } i < d-1, \\[1.5mm]
\alpha_{d-1}=\frac{1}{d}\sum_{a=1}^{d-1}a\theta_{a},  & \text{if }  i = d-1.
\end{cases}
\end{equation}
It holds that $D=BX_{d}B^{\dag}X_{d}^{\dag}$.
Since $\sum_{i=0}^{d-1}\alpha_{i}=0$, it follows that $B\in\mathrm{SU}(\mathcal{H}_{d})$.
Then one has $ABX_{d}B^{\dag}X_{d}^{\dag}A^{\dag}=U$.
By the definition of $\mathrm{C}^{2}_{1,2}(U)$, it can be decomposed as
\begin{equation}\label{eq10}
\begin{split}
\mathrm{C}^{2}_{1,2}(U)=&\; I_d \otimes (AB) \cdot \mathrm{C}^{2}_{1,2}(X_d) \\
& \cdot I_d \otimes B^{\dag} \cdot \mathrm{C}^{2}_{1,2}(X_{d}^{\dagger}) \cdot I_d \otimes A^{\dag}.
\end{split}
\end{equation}

If $U\in \mathrm{U}(\mathcal{H}_{d})$ and $\det (U)=e^{\mathrm{i}\theta}$ where $\theta\in(-\pi,\pi]$, then $e^{-\mathrm{i}\frac{\theta}{d}}U\in\mathrm{SU}(\mathcal{H}_{d})$.
It may be observed that 
\begin{equation}\label{eq11}
\begin{split}
\mathrm{C}^{2}_{1,2}(U)=&\;\Big(e^{\mathrm{i}\frac{\theta}{d}}|d-1\rangle\langle d-1|+\sum_{a=0}^{d-2}|a\rangle\langle a|\Big) \otimes I_d \\&
\cdot \mathrm{C}^{2}_{1,2}(e^{-\mathrm{i}\frac{\theta}{d}}U).
\end{split}
\end{equation}
In fact, when the control (first) qudit is in the state $|d-1\rangle$, the operator on the right-hand side of \eqref{eq11} first applies $e^{-\mathrm{i}\frac{\theta}{d}}U$ to the target (second) qudit, and multiplies the control qudit by a phase $e^{\mathrm{i}\frac{\theta}{d}}$; consequently, the overall action is to apply $U$ to the target qudit.
When the control qudit is in any state $|a\rangle$ with $a\neq d-1$, the right operator implements the identity on the target and control qudits.
Thus, Eq.~\eqref{eq7} follows from applying \eqref{eq10} to $\mathrm{C}^{2}_{1,2}(e^{-\mathrm{i}\frac{\theta}{d}}U)$ and setting
\begin{equation}\label{eq12} 
C=e^{\mathrm{i}\frac{\theta}{d}}|d-1\rangle\langle d-1|+\sum_{a=0}^{d-2}|a\rangle\langle a|.
\end{equation}

Finally, we prove the ``only if'' part of the second statement. Given $U\in \mathrm{U}(\mathcal{H}_{d})$, assume that \eqref{eq10} holds for some $A,B\in \mathrm{SU}(\mathcal{H}_{d})$. Then one has $ABX_{d}B^{\dag}X_{d}^{\dag}A^{\dag}=U$. Since $\det(ABX_{d}B^{\dag}X_{d}^{\dag}A^{\dag})=1$, it follows that $U\in \mathrm{SU}(\mathcal{H}_{d})$.
\end{proof}

Note that when $d=2$ (i.e., the qubit case), Lemma~\ref{Lem3.1} reduces to Lemma 5.1 of Ref.~\cite{Barenco1995}.
It follows from the proof that the single-qudit gates $A$ and $A^{\dag}$ in Fig.~\ref{Fig.3} can be eliminated when $U$ is a diagonal operator. 
Although we focus here on the bipartite state space $\mathcal{H}_{d}\otimes\mathcal{H}_{d}$ with equal dimensions, Lemma.~\ref{Lem3.1} actually holds for $\mathcal{H}_{d_{1}}\otimes\mathcal{H}_{d_{2}}$ even when $d_{1}\neq d_{2}$.
Moreover, $\mathrm{C}^{2}_{1,2}(X_{d}^{\dagger})$ in Fig.~\ref{Fig.3} can be transformed into $\mathrm{C}^{2}_{1,2}(X_{d})$ using the relation shown in Fig.~\ref{2a}.

\subsection{SUM-based Quantum Circuits for Single-Controlled Unitaries}\label{sec3.2}

The following lemma shows that any single-controlled qudit gate can be decomposed into $d$ SUM gates when $d$ is prime.

\begin{lemma}\label{Lem3.2}
Let $d$ be a prime number.
For $U\in \mathrm{U}(\mathcal{H}_{d})$ with $\det (U)=e^{\mathrm{i}\theta}$ for some $\theta\in(-\pi,\pi]$,
let $U=ADA^{\dag}$ where $A\in \mathrm{SU}(\mathcal{H}_{d})$ and $D$ is diagonal, and let $B\in \mathrm{SU}(\mathcal{H}_{d})$ be diagonal such that $B^d=e^{-\mathrm{i}\frac{\theta}{d}}D$.
Then $\mathrm{C}^{2,0}_{1,2}(U)$ can be decomposed into
\begin{equation}\label{eq13}
\mathrm{C}^{2,0}_{1,2}(U)= C \otimes A \cdot [\mathrm{SUM} \cdot (I_d \otimes B)]^d \cdot I_d \otimes A^{\dag},
\end{equation}
where $C$ is given by Eq.~\eqref{eq12}. The quantum circuit is shown in Fig.~\ref{Fig.4}. 
\end{lemma}

\begin{figure} [htbp]
  \centering
  \includegraphics[width=7.5cm]{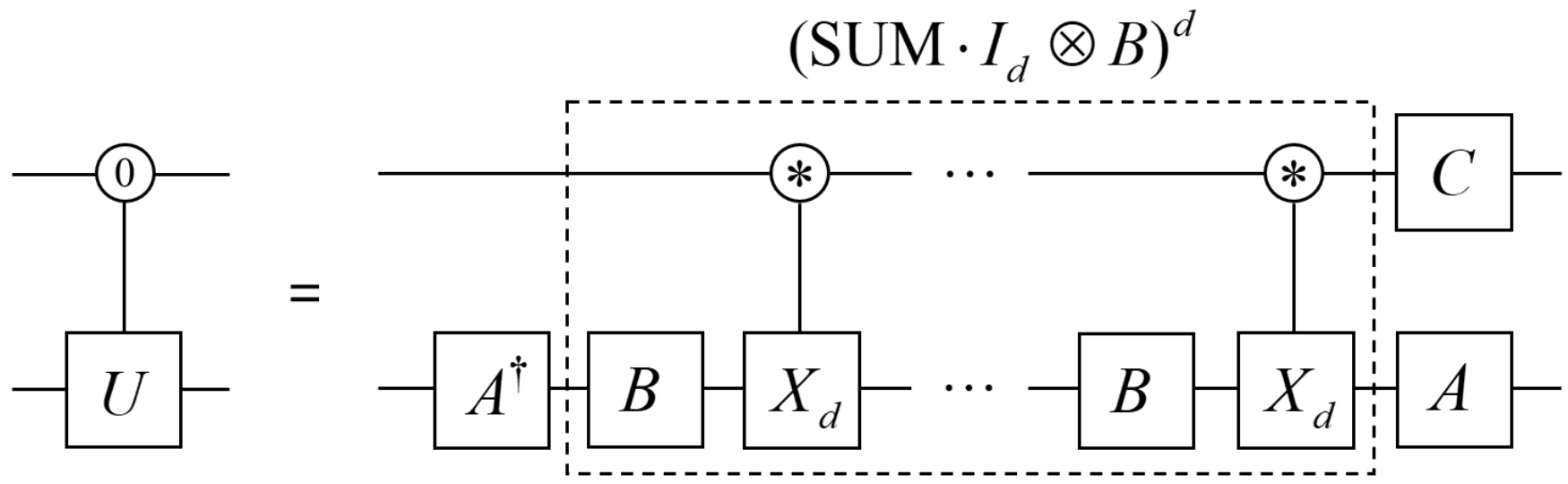}  
  \caption{Equivalent quantum circuit for the controlled $U$ gate based on Eq.~\eqref{eq13}.  }
  \label{Fig.4}
\end{figure}

\begin{proof}
First, from $U=ADA^{\dag}$ and Eq.~\eqref{eq11}, one has that 
\begin{equation}\label{eq14}
\mathrm{C}^{2,0}_{1,2}(U)= C \otimes A \cdot \mathrm{C}^{2,0}_{1,2}(e^{-\mathrm{i}\frac{\theta}{d}}D) \cdot I_d \otimes A^{\dag}.
\end{equation}
It suffices to prove that $[\mathrm{SUM} \cdot (I_d \otimes B)]^d = \mathrm{C}^{2,0}_{1,2}(e^{-\mathrm{i}\frac{\theta}{d}}D)$.
Define $B_{+n}=X_{d}^{n} B (X_{d}^{n})^{\dag}$ for every non-negative integer $n$.
It holds that 
\begin{equation}\label{eq15}
\begin{split}
(X_{d}^{a}B)^d
&=X_{d}^{a} B (X_{d}^{a}B)^{d-1}\\
&=B_{+a} X_{d}^{a} (X_{d}^{a}B)^{d-1}\\
&=B_{+a} X_{d}^{2a} B (X_{d}^{a}B)^{d-2}\\
&=B_{+a} B_{+2a} X_{d}^{2a} (X_{d}^{a}B)^{d-2} \\
&=B_{+a} B_{+2a} \ldots B_{+da} X_{d}^{da}\\
&=B_{+a} B_{+2a} \ldots B_{+da},
\end{split}
\end{equation}
where the last equality follows from $X_{d}^{d}=I_{d}$.

Next, we prove that 
\begin{equation}\label{eq16}
B_{+a} B_{+2a} \ldots B_{+da}= 
\begin{cases}
e^{-\mathrm{i}\frac{\theta}{d}}D, & \text{if } a = 0,   \\[1mm]
I_d,                              & \text{if } a \in \{1,2,\ldots,d-1\}.   \\
\end{cases}
\end{equation}
For the case of $a=0$, that follows directly from $B^d=e^{-\mathrm{i}\frac{\theta}{d}}D$.
For the case of $a \in \{1,2,\ldots,d-1\}$, let $[a]=\{a+kd : k \text{ is integer}\}$ (i.e., $[a]$ denotes a congruence class modulo $d$), and then one has  
\begin{equation}\label{eq17}
\{[a],[2a],\ldots,[da]\}=\{[0],[1],\ldots,[d-1]\}.
\end{equation}
To verify this, note that $[da]=[0]$.
For $k \in \{1,2,\ldots,d-1\}$, if $[ka]=[0]$, then $d\;|\;ka$. 
Since $d$ is prime and $d \nmid a$, one must have $d\;|\;k$, which contradicts $k \in \{1,2,\ldots,d-1\}$.
In addition, if $[ka]=[la]$ for $1 \leq k<l \leq d-1$, then $d\;|\;(l-k)a$, which again contradicts $l-k \in \{1,2,\ldots,d-2\}$.
Hence, the set $\{[a],[2a],\ldots,[da]\}$ consists of $d$ distinct congruence classes modulo $d$.
Thus, in the case of $a \in \{1,2,\ldots,d-1\}$, one has
\begin{equation}\label{eq18}
\begin{split}
B_{+a} B_{+2a} \ldots B_{+da}&=B_{+1} B_{+2} \ldots B_{+d}\\
&=\det(B)I_d=I_d,  
\end{split}
\end{equation}
where the first equation follows from the fact that $B_{+a}=B_{+b}$ if $[a]=[b]$, the second equation is derived from simple matrix calculations, and the third equation follows from $B\in \mathrm{SU}(\mathcal{H}_{d})$.

From Eqs. \eqref{eq15} and \eqref{eq16}, for every $x_1 x_2 \in [\underline{d}]^{2}$, it holds that 
\begin{equation}\label{eq19}
\begin{split}
&[\mathrm{SUM} \cdot (I_d \otimes B)]^d |x_1 x_2\rangle \\
&= |x_1\rangle \otimes (X_{d}^{x_1}B)^d|x_2\rangle\\
&= |x_1\rangle \otimes (B_{+x_1} B_{+2x_1} \ldots B_{+dx_1}) |x_2\rangle\\
&= 
\begin{cases}
|x_1\rangle \otimes e^{-\mathrm{i}\frac{\theta}{d}}D |x_2\rangle, & \text{if } x_1 = 0,   \\[1mm]
|x_1\rangle \otimes I_d |x_2\rangle,                              & \text{if } x_1 >0,
\end{cases}\\
&= \mathrm{C}^{2,0}_{1,2}(e^{-\mathrm{i}\frac{\theta}{d}}D)|x_1 x_2\rangle,
\end{split}
\end{equation}
which completes the proof.
\end{proof}

\section{Quantum Circuits for Multi-controlled Single-qudit Gates}\label{sec4}

\subsection{Quantum Circuits for 2-Controlled Single-Qudit Gates}\label{sec4.1}

The following lemma extends the result of Lemma~\ref{Lem3.1} to 2-controlled single-qudit gates, with its proof explicitly constructing the corresponding quantum circuit.

\begin{lemma}\label{Lem4.1}
For every $U\in \mathrm{U}(\mathcal{H}_{d})$, the 2-controlled single-qudit gate $\mathrm{C}^{3}_{(1,2),3}(U)$ can be synthesized using at most six CINC gates and eight single-qudit gates. 
When $U\in \mathrm{SU}(\mathcal{H}_{d})$, four CINC gates and five single-qudit gates are sufficient to synthesize $\mathrm{C}^{3}_{(1,2),3}(U)$.
\end{lemma}

\begin{proof}
For any $U\in \mathrm{U}(\mathcal{H}_{d})$ with $\det (U)=e^{\mathrm{i}\theta}$, it holds that $e^{-\mathrm{i}\frac{\theta}{d}}U \in \mathrm{SU}(\mathcal{H}_{d})$.
According to the proof of Lemma~\ref{Lem3.1}, there exist $A\in \mathrm{SU}(\mathcal{H}_{d})$ and a diagonal operator $B_0\in \mathrm{SU}(\mathcal{H}_{d})$ such that $e^{-\mathrm{i}\frac{\theta}{d}}U=AB_{0}X_{d}B_{0}^{\dag}X_{d}^{\dag}A^{\dag}$.
Recall that $C$ is the diagonal operator defined in \eqref{eq12}.
The quantum circuit in the middle of Fig.~\ref{Fig.5} faithfully simulates $\mathrm{C}^{3}_{(1,2),3}(U)$.
To verify the validity of the construction, the action of the circuit can be examined on the computational basis state $|x_{1}x_{2}x_{3}\rangle$ where $x_{1}x_{2}x_{3}\in [\underline{d}]^{3}$.
Specifically, when $|x_{1}x_{2}\rangle=|d-1\rangle^{\otimes 2}$, the quantum circuit in the middle of Fig.~\ref{Fig.5} performs the gate $U$ on the third qudit; otherwise, it acts as the identity operator.

For the gate $\mathrm{C}^{3}_{2,3}(B_{0})$ in Fig.~\ref{Fig.5}, since $B_{0}$ is a diagonal and special unitary, it follows from Lemma~\ref{Lem3.1} that
\begin{equation}\label{eq20}
\mathrm{C}^{3}_{2,3}(B_{0})=I_{d^2} \otimes B\cdot\mathrm{C}^{3}_{2,3}(X_d) \cdot I_{d^2} \otimes B^{\dag} \cdot \mathrm{C}^{3}_{2,3}(X_{d}^{\dag}),
\end{equation}
where $B\in \mathrm{SU}(\mathcal{H}_{d})$ is diagonal.
Also observe that 
\begin{equation}\label{eq21} 
\begin{split}
\mathrm{C}^{3}_{2,3}(B_{0}^{\dag})&=\big(\mathrm{C}^{3}_{2,3}(B_{0})\big)^{\dag}\\
&=\mathrm{C}^{3}_{2,3}(X_d) \cdot
I_{d^{2}} \otimes B \cdot
\mathrm{C}^{3}_{2,3}(X_{d}^{\dag}) \cdot
I_{d^{2}} \otimes B^{\dag}.
\end{split}
\end{equation}
Substituting Eqs. \eqref{eq20} and \eqref{eq21} into the middle circuit in Fig.~\ref{Fig.5}, and using the relation 
\begin{equation}\label{eq22}
\mathrm{C}^{3}_{2,3}(X_{d}^{\dag})\cdot\mathrm{C}^{3}_{1,3}(X_{d})\cdot\mathrm{C}^{3}_{2,3}(X_{d})
=\mathrm{C}^{3}_{1,3}(X_{d}),
\end{equation}
yields the right circuit in Fig.~\ref{Fig.5} that exactly simulates $\mathrm{C}^{3}_{(1,2),3}(U)$.

Finally, from Lemma~\ref{Lem3.1}, the gate $\mathrm{C}^{3}_{1,2}(C)$ in Fig.~\ref{Fig.5} can be omitted if $U\in \mathrm{SU}(\mathcal{H}_{d})$; otherwise, it can be decomposed into two CINC gates and three single-qudit gates.
This completes the proof.
\end{proof}

\begin{figure*} [htbp]
  \centering
  \includegraphics[width=15cm]{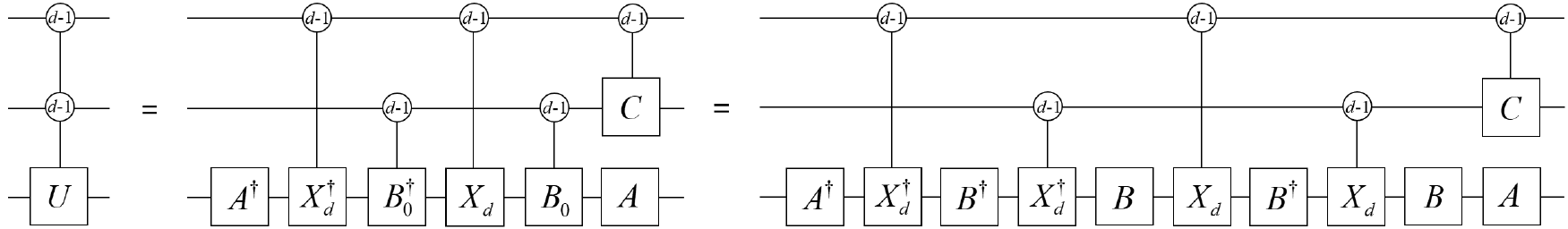}
  \caption{Equivalent quantum circuit for the 2-controlled $U$ gate.}
  \label{Fig.5}
\end{figure*}

In particular, when $d=2$, Lemma~\ref{Lem4.1} is consistent with Corollary 6.2 of Ref.~\cite{Barenco1995}.
The gates $A$ and $A^{\dag}$ in Fig.~\ref{Fig.5} can also be eliminated if $U$ is a diagonal operator.
Similar to Lemma~\ref{Lem3.1}, the conclusions of Lemma~\ref{Lem4.1} can also be extended to composite quantum systems with subsystems having different dimensions.

Since $\det(X_d)=(-1)^{d-1}$, Lemma~\ref{Lem4.1} implies that a 2-controlled INC gate $\mathrm{C}^{3}_{(1,2),3}(X_d)$ can be synthesized using four CINC gates if $d$ is odd, but six CINC gates if $d$ is even.
To eliminate this parity dependence, a pseudo-INC gate $\widetilde{X}_d$ is introduced as 
\begin{equation}\label{eq23}
\widetilde{X}_d=
\begin{cases}
X_d,   &  \text{if } d \text{ is odd},   \\[1.5mm]
e^{\mathrm{i}\frac{\pi}{d}}X_d,  &  \text{if } d \text{ is even}, 
\end{cases}
\end{equation}
which satisfies $\det(\widetilde{X}_d)=1$.
Hence, a 2-controlled pseudo-INC gate can be synthesized by four CINC gates regardless of whether $d$ is odd or even, which will be useful for simplifying quantum circuits for multi-controlled single-qudit gates.

\begin{remark}
The right-side circuit in Fig.~\ref{Fig.5} still implements the desired gate if both $\mathrm{C}^{3}_{1,3}(X_d)$ and $\mathrm{C}^{3}_{1,3}(X_d^{\dag})$ are simultaneously replaced by $\mathrm{C}^{3}_{1,3}(\widetilde{X}_d)$ and $\mathrm{C}^{3}_{1,3}(\widetilde{X}_d^{\dag})$, respectively.
The same holds if $\mathrm{C}^{3}_{2,3}(X_d)$ and $\mathrm{C}^{3}_{2,3}(X_d^{\dag})$ are simultaneously replaced by $\mathrm{C}^{3}_{2,3}(\widetilde{X}_d)$ and $\mathrm{C}^{3}_{2,3}(\widetilde{X}_d^{\dag})$, respectively.
\end{remark}

\subsection{Quantum Circuits for $(n-1)$-Controlled Single-Qudit Gates}\label{sec4.2}

We now explain how to decompose multi-controlled single-qudit gates into CINC and single-qudit gates.
The next lemma will reduce this task to the decomposition of multi-controlled pseudo-INC (or INC) gates.
A similar idea for the qubit case can be found in figure 7 or Theorem 1 of Ref~\cite{Vale2024}, but that is restricted to special unitary operators that have real diagonal entries or real off-diagonal entries (in the computational basis).
The following lemma works for any dimension $d \geq 2$ and any unitary operator.

\begin{lemma}\label{Lem4.2}
For $n\geq 4$ and every $U\in \mathrm{U}(\mathcal{H}_{d})$, the $(n-1)$-controlled single-qudit gate $\mathrm{C}^{n}_{(1,\ldots,n-1),n}(U)$ can be synthesized by the quantum circuit shown in Fig. \ref{Fig.6}, where $A\in \mathrm{SU}(\mathcal{H}_{d})$, $B\in \mathrm{SU}(\mathcal{H}_{d})$ is a diagonal operator, and $C$ is given by Eq.~\eqref{eq12}.
Moreover, $\mathrm{C}^{n}_{(1,\ldots,n-2),n-1}(C)$ can be omitted if and only if $U\in \mathrm{SU}(\mathcal{H}_{d})$.
\end{lemma}

\begin{figure*} [htbp]
  \centering
  \includegraphics[width=11cm]{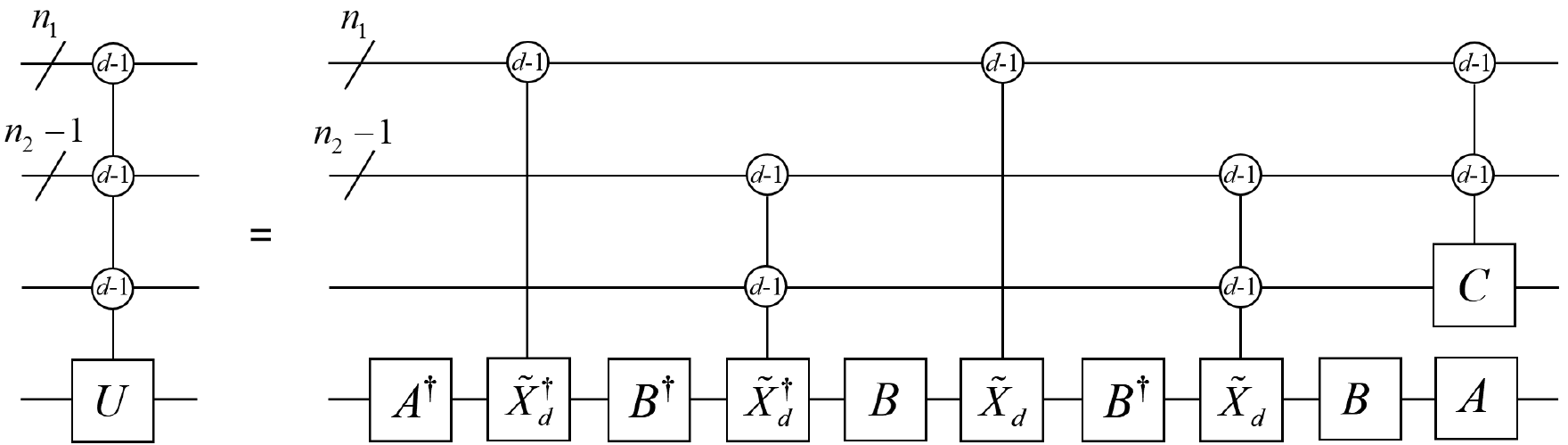}
  \caption{Equivalent quantum circuit for the $(n-1)$-controlled $U$ gate. 
  The line with a slash (/) represents a quantum bus (i.e., a collection of multiple qudits); the numbers $n_1=\lceil\frac{n-1}{2}\rceil$ and $n_2=\lfloor\frac{n-1}{2}\rfloor$ indicate the respective qudit counts on the buses, and the circle labeled $d-1$ indicates that the gate $U$ is performed only when every qudit on the control bus is in the state $|d-1\rangle$.}
  \label{Fig.6}
\end{figure*}

\begin{proof}
Assume that $\det (U)=e^{\mathrm{i}\theta}$ where $\theta\in(-\pi,\pi]$.
Then it holds that
\begin{equation}\label{eq24}
\begin{split}
\mathrm{C}^{n}_{(1,\ldots,n-1),n}(U)=&\;\mathrm{C}^{n-1}_{(1,\ldots,n-2),n-1}(C) \otimes I_d \\
&\cdot \mathrm{C}^{n}_{(1,\ldots,n-1),n}(e^{-\mathrm{i}\frac{\theta}{d}}U).
\end{split}
\end{equation}
Let $n_1=\lceil\frac{n-1}{2}\rceil$ and $n_2=\lfloor\frac{n-1}{2}\rfloor$, and partition the $n$-qudit into one three-part composite quantum system corresponding to the Hilbert space
$\mathcal{H}_{d}^{\otimes n_1}\otimes\mathcal{H}_{d}^{\otimes n_2}\otimes\mathcal{H}_{d}$.
Then $\mathrm{C}^{n}_{(1,\ldots,n-1),n}(e^{-\mathrm{i}\frac{\theta}{d}}U)$ can be viewed as a gate that is controlled by the first two subsystems (i.e., $\mathcal{H}_{d}^{\otimes n_1} \otimes \mathcal{H}_{d}^{\otimes n_2}$) and acts on the last qudit.
By Lemma~\ref{Lem4.1}, $\mathrm{C}^{n}_{(1,\ldots,n-1),n}(e^{-\mathrm{i}\frac{\theta}{d}}U)$ can be decomposed using Fig.~\ref{Fig.5}, where the far right controlled gate is omitted since $e^{-\mathrm{i}\frac{\theta}{d}}U\in\mathrm{SU}(\mathcal{H}_{d})$.
Therefore, together with Eq.~\eqref{eq24}, the circuit shown in Fig.~\ref{Fig.6} is obtained.

For the second statement, if $U\in \mathrm{SU}(\mathcal{H}_{d})$, then $\det (U)=1$, which implies $C=I_d$.
The proof of the ``only if'' part is similar to that of Lemma~\ref{Lem3.1}.
\end{proof}

\begin{remark}
Similar to Fig.~\ref{Fig.5}, the $n_1$-controlled (or $n_2$-controlled) $\widetilde{X}_d^{\dag}$ and $\widetilde{X}_d$ gates in Fig.~\ref{Fig.6} can be simultaneously replaced with $n_1$-controlled (or $n_2$-controlled) $X_d^{\dag}$ and $X_d$ gates, respectively.
\end{remark}

In particular, when $n=4$, one has $n_1=2$ and $n_2=1$.
In this case, $\mathrm{C}^{4}_{3,4}(\widetilde{X}_d^{\dag})$ and $\mathrm{C}^{4}_{3,4}(\widetilde{X}_d)$ in Fig.~\ref{Fig.6} can be simultaneously replaced with $\mathrm{C}^{4}_{3,4}(X_d^{\dag})$ and $\mathrm{C}^{4}_{3,4}(X_d)$, respectively.
From Lemma~\ref{Lem4.1}, it follows that 16 CINC gates are sufficient to synthesize the gate $\mathrm{C}^{4}_{(1,2,3),4}(U)$, while 10 CINC gates suffice when $U\in \mathrm{SU}(\mathcal{H}_{d})$.
Generally, for $U\in \mathrm{SU}(\mathcal{H}_{d})$, Lemma~\ref{Lem4.2} can be applied recursively to obtain a circuit for $\mathrm{C}^{n}_{(1,\ldots,n-1),n}(U)$ that consists of CINC and single-qudit gates. 
Such a circuit requires the following number of CINC gates:
\begin{equation}\label{eq25}
N^{\mathrm{rec}}_{\mathrm{SU}}(n)=2^{2k}+3\cdot 2^{k}(n-2^{k}-1),
\end{equation}
where $k=\lfloor\log_2(n-1)\rfloor$ and $n\geq3$.
However, $N^{\mathrm{rec}}_{\mathrm{SU}}(n)$ grows non-linearly with $n$.
In the remainder of this subsection, we present a method that achieves a linear growth in the number of CINC gates for $U \in \mathrm{SU}(\mathcal{H}_{d})$, and a quadratic scaling in the general case. 
The following lemma provides a key technique to achieve this goal, which is inspired by Lemma III.2 of Ref.~\cite{Multi-qudit2023}.
For convenience, we use the circuit symbol shown in Fig.~\ref{Fig.7}.

\begin{figure}[htbp]
  \centering
  \includegraphics[width=6.5cm]{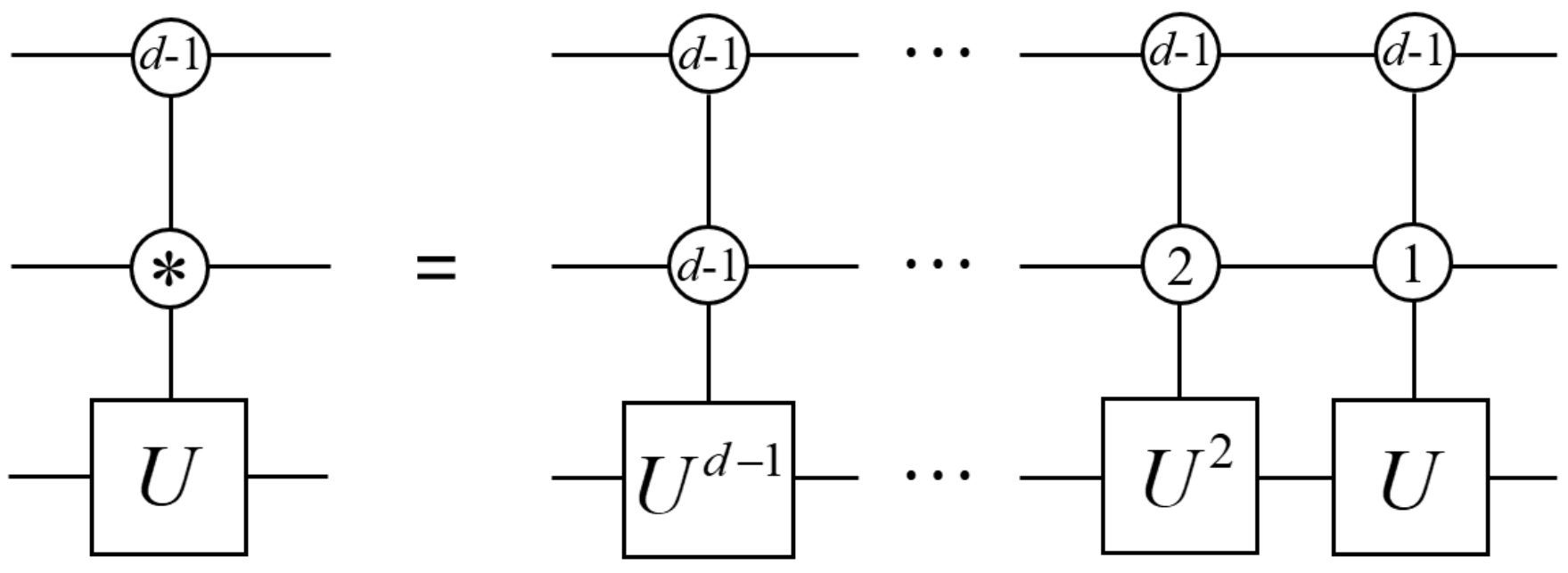}
  \caption{Quantum circuit symbol for the gate that applies $U^a$ to the target qudit when the two control qudits are in  $|d-1,a\rangle$, for $a=1,\ldots,d-1$.}\label{Fig.7}
\end{figure}

\begin{lemma}\label{Lem4.3}
For $n \geq 5$ and $m\in \{3,\dots,\lceil n/2 \rceil\}$, the gate $\mathrm{C}^{n}_{(1,\dots,m),n}(\widetilde{X}_d)$ can be synthesized using at most $(16d-20)m-40d+58$ CINC gates, up to $\mathrm{C}^{n}_{(1,\dots,m),n-1}(P_d)$, where
\begin{equation}\label{eq26}
P_d=(-1)^{d-1}|d-1\rangle\langle d-1| +\sum^{d-2}_{a=0}|a\rangle\langle a|.
\end{equation}
\end{lemma}

\begin{proof}
Fig.~\ref{Fig.8} illustrates an example for $n=9$ and $m=5$; the general case can be generalized analogously.
For every computational state $|x_{1} \dots x_{9}\rangle$ with $x_{1}\dots x_{9}\in [\underline{d}]^{9}$, the action of $\mathrm{C}^{9}_{(1,\dots,5),9}(\widetilde{X}_d) \cdot \mathrm{C}^{9}_{(1,\dots,5),8}(P_d)$ leaves the state $|x_{1} \dots x_{8}\rangle$ unchanged and applies the following operator to $|x_{9}\rangle$:
\begin{equation}\label{eq27}
\begin{cases}
I_d,  &  \text{if } |x_{1} \dots x_{5}\rangle \neq |d-1\rangle^{\otimes 5},\\
\widetilde{X}_d,   &  \text{if } |x_{1} \dots x_{5}\rangle = |d-1\rangle^{\otimes 5} , x_{8}  \neq d-1,   \\
(-1)^{d-1}\widetilde{X}_d,  &  \text{if } |x_{1} \dots x_{5}\rangle = |d-1\rangle^{\otimes 5}, x_{8}  = d-1. 
\end{cases}
\end{equation}
Next we show that the right circuit in Fig.~\ref{Fig.8} simulates this operation.

We first consider the action of the first seven gates.
Assume that $|x_{8}\rangle \neq |d-1\rangle$ (resp. $|x_{8}\rangle = |d-1\rangle$).
One may observe that $\widetilde{X}_d$ (resp. $(\widetilde{X}_d^{\dag})^{d-1}=(-1)^{d-1}\widetilde{X}_d$) is applied to the ninth qudit if and only if both of the following hold: (i) $|x_{5}\rangle=|d-1\rangle$ and (ii) $|x_{8}\rangle$ is mapped to $|(x_{8}+1)\bmod d\rangle$ up to a phase.
Condition (ii) holds if and only if both of the following hold: (iii) $|x_{4}\rangle=|d-1\rangle$ and (iv) $|x_{7}\rangle$ is mapped to $|(x_{7}+1)\bmod d\rangle$ up to a phase.
Condition (iv) holds if and only if both of the following hold: (v) $|x_{3}\rangle=|d-1\rangle$ and (vi) $|x_{6}\rangle$ is mapped to $|(x_{6}+1)\bmod d\rangle$ up to a phase.
Condition (vi) holds if and only if $|x_{1}x_{2}\rangle=|d-1,d-1\rangle$.
Thus, in the case of $|x_{8}\rangle \neq |d-1\rangle$ (resp. $|x_{8}\rangle = |d-1\rangle$), $\widetilde{X}_d$ (resp. $(-1)^{d-1}\widetilde{X}_d$) is applied to the ninth qudit if and only if $|x_{1} x_{2} x_{3} x_{4} x_{5}\rangle=|d-1\rangle^{\otimes 5}$.
Finally, since the two sets of gates in the dashed boxes are inverses of each other, the overall circuit leaves $|x_{1}\dots x_{8}\rangle$ unchanged.
Therefore the circuit simulates the desired operation.

By Lemma~\ref{Lem4.1} and Fig.~\ref{Fig.7}, $16(d-1)m-40d+48$ CINC gates are required for this synthesis.
However, $4m-10$ CINC gates can be further omitted, which yields the claimed gate count.
Figure~\ref{Fig.9} presents an example of this reduction.
Applying the decomposition shown in Fig.~\ref{Fig.5} to the 2-controlled $\widetilde{X}_d^{\dag}$ and $\widetilde{X}_d$ gates in the left circuit of Fig.~\ref{Fig.9} yields the right circuit of Fig.~\ref{Fig.9}, where the gates in the dashed boxes cancel each other.
Hence, each gate with an asterisk ($*$) in Fig.~\ref{Fig.8} saves one CINC gate.
\end{proof}

\begin{figure*}[htbp]
  \centering
  \includegraphics[width=12cm]{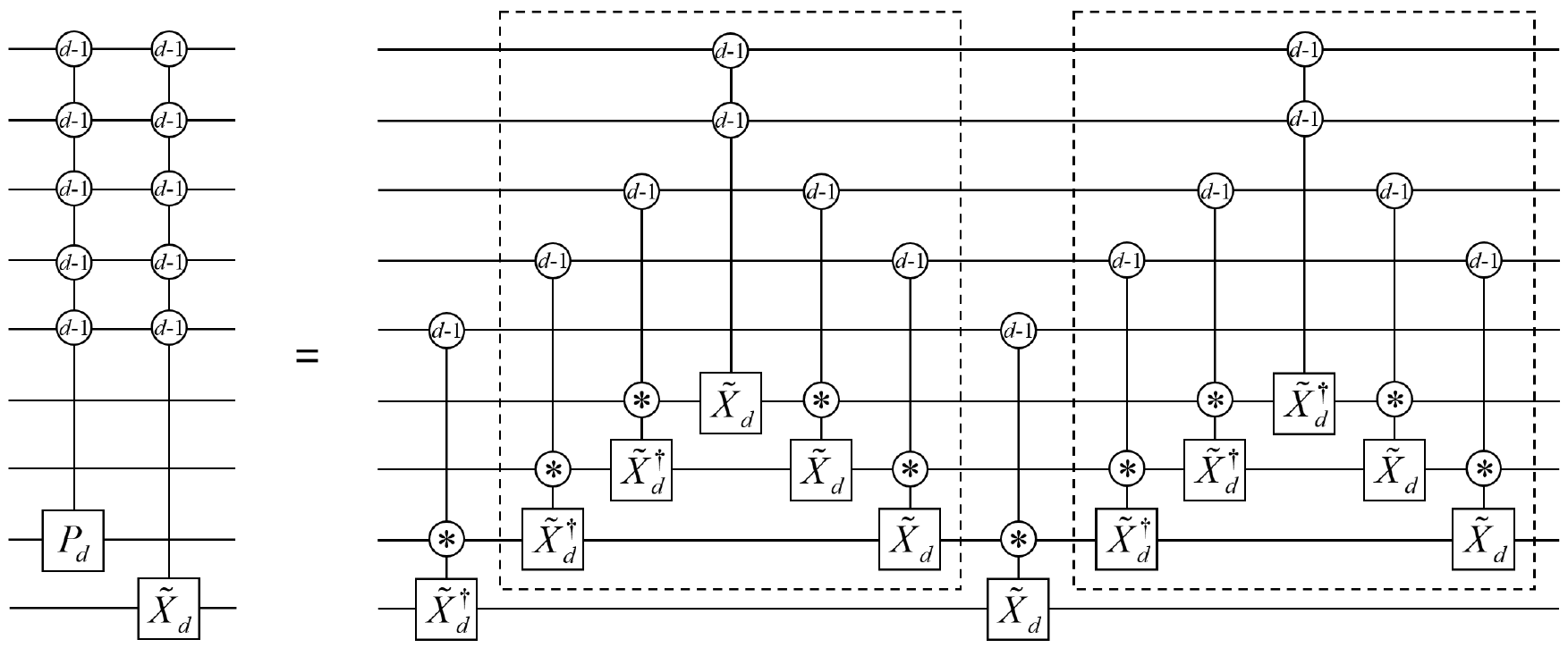}
  \caption{Equivalent quantum circuit for $\mathrm{C}^{9}_{(1,\dots,5),9}(\widetilde{X}_d) \cdot \mathrm{C}^{9}_{(1,\dots,5),8}(P_d)$.}\label{Fig.8}
\end{figure*}

\begin{figure*}[htbp]
  \centering
  \includegraphics[width=16cm]{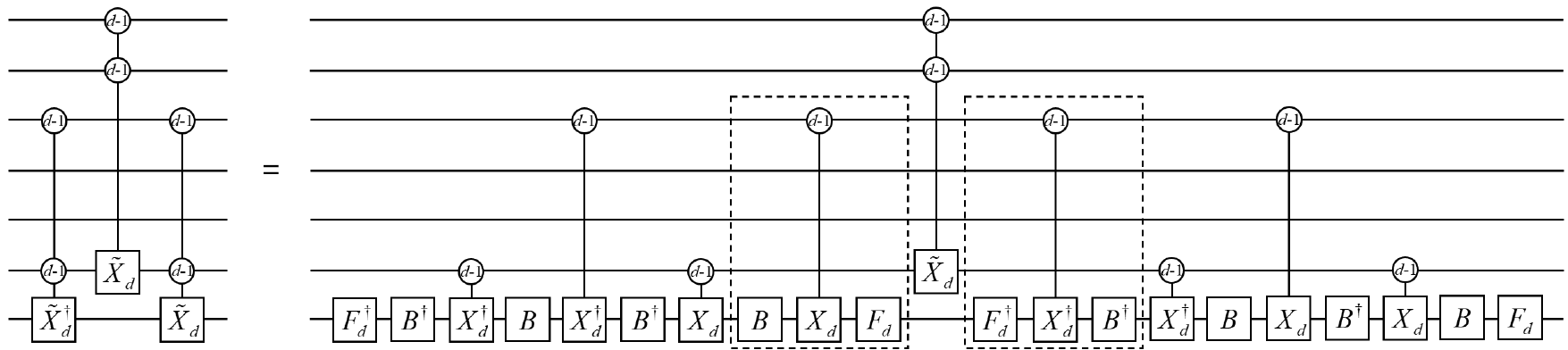}
  \caption{Illustration of gate cancellation, where $F_d$ is given by Eq.~\eqref{eq1}.}\label{Fig.9}
\end{figure*}

When $d=2$ or $d=3$, one can apply Lemma 8 in \cite{isometries} and Lemma 1 in \cite{qutrit2022} to further reduce the number of CINC gates required to synthesize an $m$-controlled INC gate under the conditions of Lemma~\ref{Lem4.3}.
The following two lemmas present the corresponding results for $d=2$ and $d=3$, respectively.

\begin{lemma}\label{Lem4.4}
For $n\geq 5$ and $m\in \{3,\dots,\lceil n/2 \rceil\}$, the $m$-controlled $X_2$ gate $\mathrm{C}^{n}_{(1,\dots,m),n}(X_2)$ can be synthesized using at most $8m-8$ CNOT gates.
\end{lemma}

\begin{proof}
The proof is identical to that of Lemma 8 in \cite{isometries}, except that the 2-controlled NOT gates are decomposed using Fig.~\ref{Fig.5} rather than the approach employed in \cite{isometries}.
\end{proof}

\begin{lemma}\label{Lem4.5}
For $n\geq 5$ and $m\in \{3,\ldots,\lceil n/2 \rceil\}$, the $m$-controlled $X_3$ gate $\mathrm{C}^{n}_{(1,\ldots,m),n}(X_3)$ can be synthesized using at most $16m-32$ CINC gates.
\end{lemma}

\begin{proof} 
The proof is similar to that of Lemma~\ref{Lem4.3}, except that the gates with an asterisk ($*$) in Fig.~\ref{Fig.8} are decomposed using Fig.~\ref{Fig.10} (a variant of Lemma 1 in \cite{qutrit2022}).
\end{proof}

\begin{figure}[htbp]
  \centering
  \includegraphics[width=6.5cm]{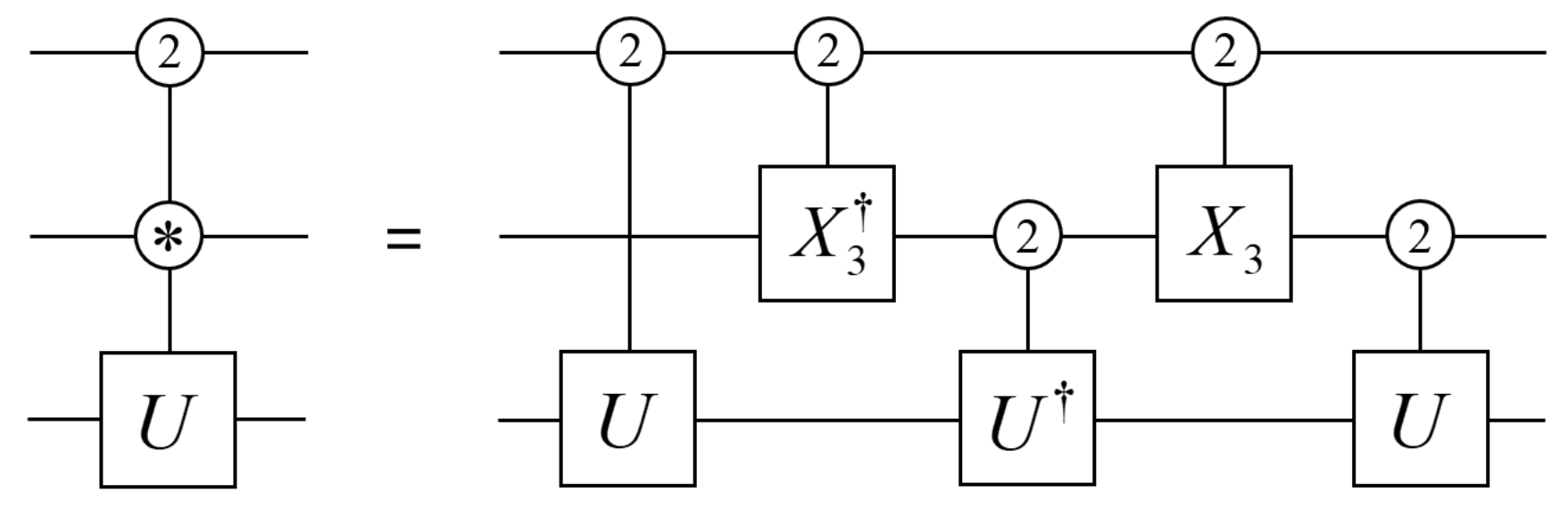}
  \caption{Equivalent 3-qutrit circuit based on Lemma 1 of \cite{qutrit2022}, where $U\in\mathrm{U}(\mathcal{H}_{3})$.}\label{Fig.10}
\end{figure}

Combining Lemma~\ref{Lem4.2} with Lemmas~\ref{Lem4.3} to \ref{Lem4.5}, for a fixed $d$, one obtains an upper bound on the number of CINC gates required to synthesize an $(n-1)$-controlled special unitary gate, which scales linearly with $n$.

\begin{theorem}\label{The4.6}
For $n\geq 7$ and $U\in \mathrm{SU}(\mathcal{H}_{d})$, the number of CINC gates required to synthesize  $\mathrm{C}^{n}_{(1,\dots,n-1),n}(U)$ is at most 
\begin{equation}\label{eq28}
N_{\mathrm{SU}}(d,n)=
\begin{cases}
16n-48,              & \text{if } d = 2,   \\
32n-160,             & \text{if } d = 3,   \\
(32d-40)n-192d+272,  & \text{if } d \geq 4.     \\
\end{cases}
\end{equation}
\end{theorem}

\begin{proof}
First, the gate $\mathrm{C}^{n}_{(1,\dots,n-1),n}(U)$ can be decomposed via Lemma~\ref{Lem4.2}.
Each multi-controlled $\widetilde{X}_d^{\dag}$ gate in Fig.~\ref{Fig.6} is equivalent to a multi-controlled $\widetilde{X}_d$ gate up to two $T_d$ gates as shown in Fig.~\ref{2a}.
For $d=2$, by replacing all multi-controlled $\widetilde{X}_2$ gates with multi-controlled $X_2$ gates, the total gate count follows from Lemma~\ref{Lem4.4}.
For $d=3$, since $\widetilde{X}_3=X_3$, the result follows directly from Lemma~\ref{Lem4.5}.
For the case of $d\geq4$, we can use Lemma~\ref{Lem4.3} to implement each multi-controlled $\widetilde{X}_d$ gate up to a multi-controlled $P_d$ gate.
Note that this results in the same overall action since the multi-controlled $P_d$ gates cancel each other in pairs. 
Summing the CINC counts given in Lemma~\ref{Lem4.3} yields the desired bound.
\end{proof}

Figure~\ref{Fig.11} illustrates a comparison of $N_{\mathrm{SU}}(d,n)$ for $d\in\{2,3,4,5\}$ with $N^{\mathrm{rec}}_{\mathrm{SU}}(n)$.
For large $n$, Theorem~\ref{The4.6} exhibits a clear advantage over the recursive application of Lemma~\ref{Lem4.2}.

\begin{figure}[htbp]
  \centering
  \includegraphics[width=8.5cm]{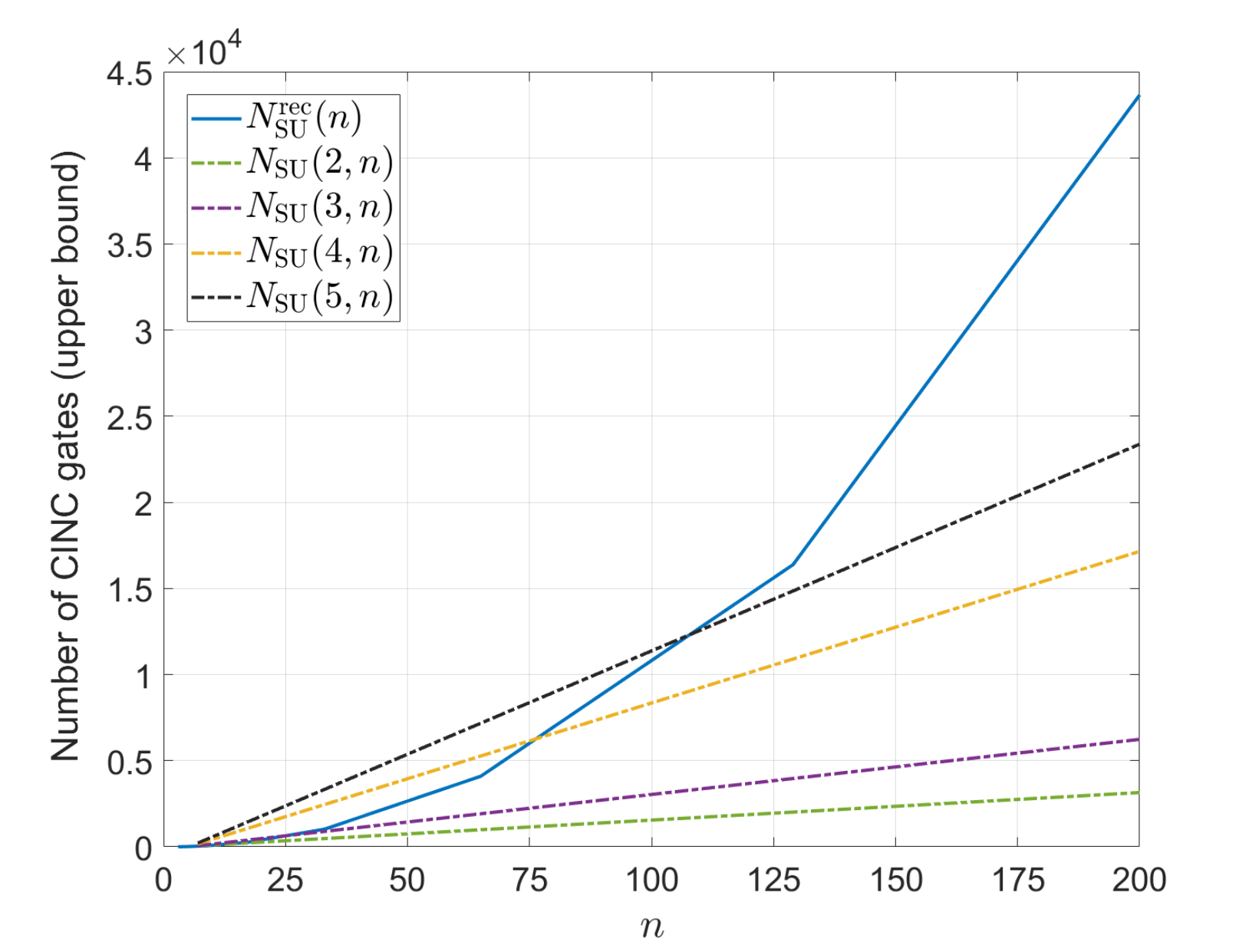}
  \caption{Comparison of the upper bounds on the number of CINC gates required to synthesize $\mathrm{C}^{n}_{(1,\dots,n-1),n}(U)$, where $U\in \mathrm{SU}(\mathcal{H}_{d})$. Here $N_{\mathrm{SU}}^{\mathrm{rec}}(n)$ and $N_{\mathrm{SU}}(d,n)$ are given by \eqref{eq25} and \eqref{eq28}, respectively.}\label{Fig.11}
\end{figure}

For $U\in \mathrm{U}(\mathcal{H}_{d})$ and a fixed $d$, by combining Theorem \ref{The4.6} with Lemma \ref{Lem4.2}, we can obtain a CINC count for synthesizing an $(n-1)$-controlled $U$ gate, which scales quadratically with $n$.

\begin{theorem}\label{The4.7}
For $n\geq 3$ and $U\in \mathrm{U}(\mathcal{H}_{d})$, the number of CINC gates required to synthesize  $\mathrm{C}^{n}_{(1,\dots,n-1),n}(U)$ is at most 
\begin{equation}\label{eq29}
N_{\mathrm{U}}(d,n)=2+\sum^{n}_{m=3}\min\{N_{\mathrm{SU}}^{\mathrm{rec}}(m),N_{\mathrm{SU}}(d,m)\},
\end{equation}
where $N_{\mathrm{SU}}^{\mathrm{rec}}(m)$ and $N_{\mathrm{SU}}(d,m)$ are given by Eqs.~\eqref{eq25} and \eqref{eq28}, respectively.
Here, for $3 \leq m \leq 6$, the minimum evaluates to $N_{\mathrm{SU}}^{\mathrm{rec}}(m)$ since $N_{\mathrm{SU}}(d,m)$ is undefined in this regime.
Moreover, when $n\geq7$, it holds that
\begin{equation}\label{eq30}
N_{\mathrm{U}}(d,n)\leq
\begin{cases}
8n^2-40n+12, & \text{if } d = 2, \\
16n^2-144n+348, & \text{if } d = 3, \\[1mm]
\begin{aligned}
&(16d-20)n^2 - (176d-252)n\\
&  + (490d-732),
\end{aligned} & \text{if } d \geq 4. 
\end{cases}
\end{equation}
\end{theorem}

\begin{proof}
From Eq.~\eqref{eq24}, it follows that the number of CINC gates for synthesizing $\mathrm{C}^{n}_{(1,\dots,n-1),n}(U)$ equals the sum of the CINC counts required for synthesizing an $(n-2)$-controlled $C$ gate and for synthesizing an $(n-1)$-controlled special unitary gate.
Therefore, we have
\begin{equation}\label{eq31}
N_{\mathrm{U}}(d,n)=N_{\mathrm{U}}(d,n-1)+\min\{N_{\mathrm{SU}}^{\mathrm{rec}}(n),N_{\mathrm{SU}}(d,n)\}.
\end{equation}
Thus Eq.~\eqref{eq29} is established by the recursive application of \eqref{eq31} and $N_{\mathrm{U}}(d,2)=2$ by Lemma~\ref{Lem3.1}.
In addition, for $n \geq 7$, Eq.~\eqref{eq30} follows directly from the inequality
\begin{equation}\label{eq32}
N_{\mathrm{U}}(d,n) \leq 2+\sum^{6}_{m=3}N_{\mathrm{SU}}^{\mathrm{rec}}(m) + \sum^{n}_{m=7}N_{\mathrm{SU}}(d,m).
\end{equation}
\end{proof}

For prime $d$, Lemma~\ref{Lem3.2} implies that each CINC gate can be decomposed into $d$ SUM gates. 
Combining this with Theorems~\ref{The4.6} and \ref{The4.7} establishes the following corollary.

\begin{corollary}\label{Cor4.8}
For $n\geq 3$ and a prime $d\geq 2$, the number of SUM gates required to synthesize an $(n-1)$-controlled special unitary gate (resp. unitary gate) is at most $dN_{\mathrm{SU}}(d,n)$ (resp. $dN_{\mathrm{U}}(d,n)$), where $N_{\mathrm{SU}}(d,n)$ and $N_{\mathrm{U}}(d,n)$ are given by Eqs. \eqref{eq28} and \eqref{eq29}, respectively.
\end{corollary}

In addition, it is possible to implement an $(n-1)$-controlled unitary gate using $O(n)$ CINC gates with the aid of an ancilla qudit.

\begin{corollary}\label{Cor4.9}
Let $U\in \mathrm{U}(\mathcal{H}_{d})$ be a unitary operator with $\det(U)=e^{\mathrm{i}\theta}\neq1$ for some $\theta\in (-\pi,\pi]\setminus\{0\}$.
For $n\geq 7$, with the aid of an ancilla qudit initialized and restored to the state $|0\rangle_{\mathrm{anc}}$, the number of CINC gates for synthesizing  $\mathrm{C}^{n}_{(1,\dots,n-1),n}(U)$ is at most $2N_{\mathrm{SU}}(d,n)$.
\end{corollary}

\begin{proof}
This result follows immediately from the circuit construction shown in Fig.~\ref{Fig.12}, where 
\begin{equation}\label{eq33}
\begin{split}
\widetilde{U}&=e^{-\mathrm{i}\frac{\theta}{d}}U,\\
\widetilde{C}&=e^{\mathrm{i}\frac{\theta}{d}} |0\rangle\langle 0| + e^{-\mathrm{i}\frac{\theta}{d}} |1\rangle\langle 1| + \sum^{d-2}_{a=2}|a\rangle\langle a|.
\end{split}
\end{equation}
Since $\widetilde{U}$ and $\widetilde{C}$ are special unitaries, both $(n-1)$-controlled gate in Fig.~\ref{Fig.12} can be synthesized using at most $N_{\mathrm{SU}}(d,n)$ CINC gates by Theorem~\ref{The4.6}.
A similar construction for qubits has appeared in figure 7 of Ref.~\cite{Rosa2025}.
\end{proof}

\begin{figure}[htbp]
  \centering
  \includegraphics[width=6.5cm]{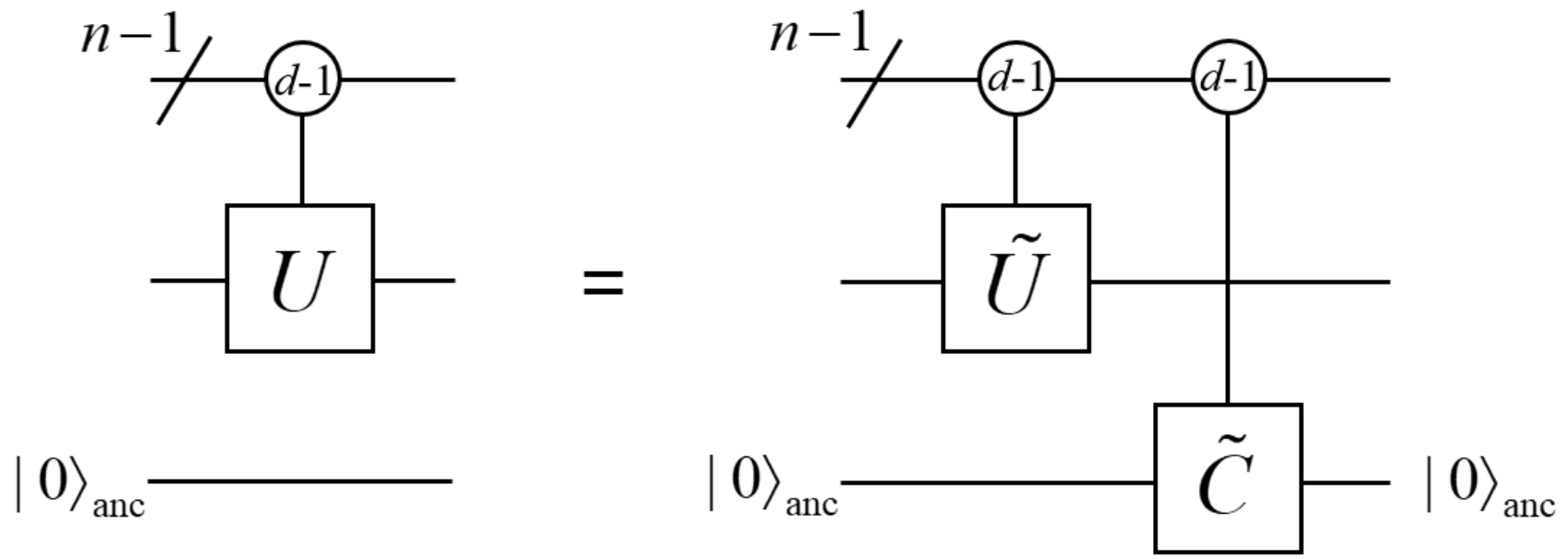}
  \caption{Equivalent quantum circuit for $\mathrm{C}^{n}_{(1,\ldots,n-1),n}(U)$ with one-qudit assistance.}\label{Fig.12}
\end{figure}

\subsection{Approximate Synthesis of Multi-Controlled Unitary Gates}\label{sec4.3}

Indeed, the decomposition of $\mathrm{C}^{n}_{(1,\dots,n-1),n}(U)$ is a recursive application of Fig.~\ref{Fig.6} that continues until the far right gate in the circuit is a single-controlled qudit gate.
For $1\leq k\leq n-2$, after the $k$-th recursive application, the far-right gate in the circuit becomes an $(n-k-1)$-controlled $C_k$ gate, where
\begin{equation}\label{eq34} 
C_k=e^{\mathrm{i}\frac{\theta}{d^{k}}}|d-1\rangle\langle d-1|+\sum_{a=0}^{d-2}|a\rangle\langle a|.
\end{equation}
For sufficiently large $n$, $C_k$ approaches the identity operator as $k$ increases.
Thus, if a small error is allowed in the circuit simulation, one can choose an appropriate $k$ and omit the $(n-k-1)$-controlled $C_k$ gate, thereby reducing the CINC gate count.

Formally, we consider the trace distance $\mathrm{T}(\rho,\sigma):=\frac{1}{2}\|\rho-\sigma\|_1$ for density operators $\rho$ and $\sigma$, where $\|\cdot\|_1$ denotes the trace norm defined as $\|A\|_1=\mathrm{Tr}(\sqrt{A^{\dag}A})$.
A natural extension of the trace distance between two unitary operators $V,W \in \mathrm{U}(\mathcal{H}_{d}^{\otimes n})$ is given by \cite{Sanders2015}
\begin{equation}\label{eq35}
\begin{split} 
\mathrm{T}(V,W):=&\max_{\rho\in \mathcal{D}(\mathcal{H}_{d}^{\otimes n})}\mathrm{T}(V \rho V^{\dagger}, W \rho W^{\dagger})\\
=&\max_{|\psi\rangle\in \mathcal{S}(\mathcal{H}_{d}^{\otimes n})}\mathrm{T}(V |\psi\rangle\langle\psi| V^{\dagger}, W |\psi\rangle\langle\psi| W^{\dagger})\\
=&\max_{|\psi\rangle\in \mathcal{S}(\mathcal{H}_{d}^{\otimes n})} \sqrt{1-| \langle\psi| V^{\dagger} W |\psi\rangle |^2}.
\end{split}
\end{equation}
We say that $V$ is an $\epsilon$-approximation of $W$ with respect to the trace distance if $\mathrm{T}(V,W) \leq \epsilon$.
The trace distance quantifies the distinguishability between the output states generated by two unitary operators for the worst-case input state.
Moreover, the \emph{gate fidelity} \cite{Nielsen2000,Huang2022} between two unitary operators $V,W \in \mathrm{U}(\mathcal{H}_{d}^{\otimes n})$ is defined as 
\begin{equation}\label{eq36}
\begin{split} 
F(V,W):=&\min_{\rho\in \mathcal{D}(\mathcal{H}_{d}^{\otimes n})}F(V \rho V^{\dagger}, W \rho W^{\dagger})\\
=&\min_{|\psi\rangle\in \mathcal{S}(\mathcal{H}_{d}^{\otimes n})}F(V |\psi\rangle\langle\psi| V^{\dagger}, W |\psi\rangle\langle\psi| W^{\dagger})\\
=&\min_{|\psi\rangle\in \mathcal{S}(\mathcal{H}_{d}^{\otimes n})} | \langle\psi| V^{\dagger} W |\psi\rangle |,
\end{split}
\end{equation}
where $F(\rho,\sigma):=\|\sqrt{\rho} \sqrt{\sigma}\|_1$ for density operators $\rho$ and $\sigma$ and the second equality follows from the joint concavity of fidelity.
By the Fuchs--van de Graaf inequalities, the trace distance and the gate fidelity satisfy
\begin{equation}\label{eq37} 
1-F(V,W) \leq \mathrm{T}(V,W).
\end{equation}
Hence, a sufficiently small trace distance guarantees a high gate fidelity.

Alternatively, one can characterize the approximation via the spectral norm $\|\cdot\|$ \cite{Barenco1995,Silva2025}. The spectral norm of an operator $A$ is defined as $\|A\| := \max_{|\psi\rangle \in \mathcal{S}(\mathcal{H}_d^{\otimes n})} \| A |\psi\rangle \|$. For two unitary operators $V, W \in \mathrm{U}(\mathcal{H}_d^{\otimes n})$, $V$ is said to be an $\epsilon$-approximation of $W$ with respect to the spectral norm if $\| V - W \| \le \epsilon$.
Moreover, the spectral norm distance is also related to the gate fidelity through the following inequality:
\begin{equation}\label{eq38}
\begin{split} 
1-F(V,W)=&\max_{|\psi\rangle\in \mathcal{S}(\mathcal{H}_{d}^{\otimes n})} \big(1- | \langle\psi| V^{\dagger} W |\psi\rangle |\big)\\
\leq & \max_{|\psi\rangle\in \mathcal{S}(\mathcal{H}_{d}^{\otimes n})} \big(1- \mathrm{Re}(\langle\psi| V^{\dagger} W |\psi\rangle )\big)\\
=&\max_{|\psi\rangle\in \mathcal{S}(\mathcal{H}_{d}^{\otimes n})} \frac{1}{2} \|(V-W)|\psi\rangle\|^2 \\
=&\frac{1}{2} \|(V-W)\|^2,
\end{split}
\end{equation}
where $\mathrm{Re}(z)$ denotes the real part of a complex number $z$.
This relation shows that the infidelity is quadratically bounded by the spectral norm error.

The following theorem establishes the synthesis results under the considered approximation criteria.

\begin{theorem}\label{The4.10}
Let $U\in \mathrm{U}(\mathcal{H}_{d})$ be a unitary operator with $\det(U)=e^{\mathrm{i}\theta}\neq1$ for some $\theta\in (-\pi,\pi]\setminus\{0\}$.
Given any error tolerance $\epsilon\in (0,1)$, let $k_1=\lceil \log_{d}\frac{|\theta|}{2\arcsin(\epsilon)} \rceil$ and $k_2=\lceil \log_{d}\frac{|\theta|}{2\arcsin(\epsilon/2)} \rceil$.
Then, the following statements hold:

(i) For $n>k_1$, there exists an $\epsilon$-approximation of $\mathrm{C}^{n}_{(1,\dots,n-1),n}(U)$ with respect to the trace distance that can be synthesized using at most $k_1 N_{\mathrm{SU}}(d,n)$ CINC gates.

(ii) For $n>k_2$, there exists an $\epsilon$-approximation of $\mathrm{C}^{n}_{(1,\dots,n-1),n}(U)$ with respect to the spectral norm that can be synthesized using at most $k_2 N_{\mathrm{SU}}(d,n)$ CINC gates.

\end{theorem}

\begin{proof}
By recursively applying the decomposition shown in Fig.~\ref{Fig.6} $k$ times, we have  
\begin{equation}\label{eq39}
\begin{split}
&\mathrm{C}^{n}_{(1,\dots,n-1),n}(U)\\
&=[\mathrm{C}^{n-1}_{(1,\dots,n-2),n-1}(C_1) \otimes I_{d}] \cdot V^{(1)}  \\
&=[\mathrm{C}^{n-2}_{(1,\dots,n-3),n-2}(C_2) \otimes I_{d^2}] \cdot (V^{(2)} \otimes I_{d}) \cdot V^{(1)}  \\
&\;\;\;\;\;\cdots\\
&=[\mathrm{C}^{n-k}_{(1,\dots,n-k-1),n-k}(C_k) \otimes I_{d^k}] \\
&\;\;\;\;\; \cdot (V^{(k)} \otimes I_{d^{k-1}} \cdot \ldots \cdot V^{(2)} \otimes I_{d} \cdot V^{(1)} ),
\end{split}
\end{equation}
where $C_k$ is given by Eq.~\eqref{eq34} and $V^{(k)}$ is an $(n-k)$-controlled special unitary gate.
For convenience, let 
\begin{equation}\label{eq40}
\begin{split}
&\mathrm{C}^{n}_{(1,\dots,n-1),n}(U)=V, \\
&V^{(k)} \otimes I_{d^{k-1}} \cdot \ldots \cdot V^{(2)} \otimes I_{d} \cdot V^{(1)}=V_k.
\end{split}
\end{equation}
Since $V_k$ commutes with $\mathrm{C}^{n-k}_{(1,\dots,n-k-1),n-k}(C_k) \otimes I_{d^{k}}$, we obtain
\begin{equation}\label{eq41}
V_k^{\dagger} V = VV_k^{\dagger} = \mathrm{C}^{n-k}_{(1,\dots,n-k-1),n-k}(C_k) \otimes I_{d^{k}}.
\end{equation}

From Eqs.~\eqref{eq41} and \eqref{eq35}, the trace distance between $V_k$ and $V$ is given by
\begin{equation}\label{eq42}
\begin{split}
&1-\mathrm{T}(V_k,V)^2\\
&=\min_{|\psi\rangle\in \mathcal{S}(\mathcal{H}_{d}^{\otimes n})}| \langle\psi| V_k^{\dagger} V |\psi\rangle |^2\\
&=\min_{|\psi\rangle\in \mathcal{S}(\mathcal{H}_{d}^{\otimes n})}| \langle\psi| \mathrm{C}^{n-k}_{(1,\dots,n-k-1),n-k}(C_k) \otimes I_{d^{k}}  |\psi\rangle |^2\\
&=\min_{p\in [0,1]}| (1-p) + p \cdot e^{i\frac{\theta}{d^k}} |^2\\
&=\min_{p\in [0,1]}[1-2p(1-p)(1-\cos\tfrac{\theta}{d^k})]\\
&=\cos^2 \tfrac{\theta}{2d^k},
\end{split}
\end{equation}
where the third equality holds because $\mathrm{C}^{n-k}_{(1,\dots,n-k-1),n-k}(C_k) \otimes I_{d^{k}}$ introduces a phase factor of $e^{\mathrm{i}\frac{\theta}{d^k}}$ to the computational basis states satisfying the control condition, while leaving the remaining basis states invariant.
Therefore, we have 
\begin{equation}\label{eq43}
\mathrm{T}(V_k,V)=\sin \tfrac{|\theta|}{2d^k}.
\end{equation}

By Eq.~\eqref{eq43}, one has $\mathrm{T}(V_k,V) \leq \epsilon$ if and only if $k\geq \log_{d}\frac{|\theta|}{2\arcsin(\epsilon)} $.
Hence, choosing $k=k_1$ guarantees that $V_{k_1}$ is an $\epsilon$-approximation of $V$ with respect to the trace distance.
From Eq.~\eqref{eq40} and Theorem~\ref{The4.6}, $V_{k_1}$ can be synthesized using at most
\begin{equation}\label{eq44}
\sum^{k_1}_{i=0}\min\{N_{\mathrm{SU}}^{\mathrm{rec}}(n-i),N_{\mathrm{SU}}(d,n-i)\} \leq k_1N_{\mathrm{SU}}(d,n).
\end{equation}
CINC gates. Thus, the statement (i) holds.

Similarly, from Eq.~\eqref{eq41} and the unitary invariance of the spectral norm, the spectral norm distance between $V_k$ and $V$ is given by
\begin{equation}\label{eq45}
\begin{split} 
\|V_k-V\|=&\|V_k(I_{d^n}-V_k^{\dag}V)\| = \|I_{d^n}-V_k^{\dag}V\|\\
=&\|I_{d^n}-\mathrm{C}^{n-k}_{(1,\dots,n-k-1),n-k}(C_k) \otimes I_{d^{k}}\|\\
=&|1-e^{\mathrm{i}\frac{\theta}{d^{k}}}|=2\sin \tfrac{|\theta|}{2d^k},
\end{split}
\end{equation}
where fourth equality follows from the fact that
$I_{d^n}-\mathrm{C}^{n-k}_{(1,\dots,n-k-1),n-k}(C_k)\otimes I_{d^k}$
is a diagonal operator whose nonzero diagonal entries are
$1-e^{\mathrm{i}\frac{\theta}{d^k}}$.

By Eq.~\eqref{eq45}, the condition $\|V_k-V\|\leq\epsilon$ is equivalent to 
$ k\geq\log_d\frac{|\theta|}{2\arcsin(\epsilon/2)}$.
Hence, choosing $k=k_2$ guarantees that $V_{k_2}$ is an $\epsilon$-approximation of $V$ with respect to the spectral norm.
Finally, by the same gate-count analysis as in statement (i), $V_{k_2}$ can be synthesized using at most $k_2N_{\mathrm{SU}}(d,n)$ CINC gates.
\end{proof}

According to the proof of Theorem~\ref{The4.10}, the spectral norm error between $V$ and $V_k$ is exactly twice the trace distance:
\begin{equation}\label{eq46} 
\|V_k-V\|=2\mathrm{T}(V_k,V)=2\sin \tfrac{|\theta|}{2d^k}.
\end{equation}
Here, $V$ denotes the ideal multi-controlled unitary gate $\mathrm{C}^{n}_{(1,\dots,n-1),n}(U)$, and $V_k$ denotes an approximation of $V$ obtained by recursively applying the decomposition shown in Fig.~\ref{Fig.6} $k$ times and omitting the far-right $(n-k-1)$-controlled $C_k$ gate.
Since the approximation error decreases as the local dimension $d$ increases, higher-dimensional quantum systems require fewer recursive decomposition steps to achieve a prescribed approximation accuracy, leading to a lower synthesis cost in the proposed method.
Table~\ref{Tab:Error} lists the spectral norm errors $\|V_k-V\|$ for different local dimensions $d$ and recursive decomposition steps $k$, where the case of $|\theta|=\pi$ (i.e., $\det(U)=-1$) is considered.
For example, when $d=2$, the error is reduced to approximately 0.0491 after six recursive decomposition steps, while for $d=5$, the error is reduced to approximately 0.0002 after the same number of steps.

\begin{table}[htbp]
\centering
\caption{Spectral norm errors between the approximation $V_k$ and the ideal multi-controlled unitary gate $V$ for different local dimensions $d$ and recursive decomposition steps $k$.}
\label{Tab:Error}
\begin{tabular}{c@{\qquad} c @{\quad} c @{\quad} c @{\quad} c @{\quad} c @{\quad} c }
\toprule
 & \multicolumn{6}{c}{\(k\)}\\
\cmidrule(lr){2-7}
$d$ & $1$ & $2$ & $3$ & $4$ & $5$ & $6$ \\
\midrule
$2$ & 1.4142 & 0.7654 & 0.3902 & 0.1960 & 0.0981 & 0.0491 \\
$3$ & 1.0000 & 0.3473 & 0.1163 & 0.0388 & 0.0129 & 0.0043 \\
$4$ & 0.7654 & 0.1960 & 0.0491 & 0.0123 & 0.0031 & 0.0008 \\
$5$ & 0.6180 & 0.1256 & 0.0251 & 0.0050 & 0.0010 & 0.0002 \\
\bottomrule
\end{tabular}
\end{table}

Theorem~\ref{The4.10} characterizes the approximation error of a single multi-controlled unitary gate.
In some quantum circuits or algorithms, however, multiple multi-controlled unitary gates may be required.
In this case, the total approximation error can be bounded by the sum of the individual approximation errors.
Consider a circuit consisting of $m$ unitary gates $U=U_1 U_2 \ldots U_m$.
If each $U_i$ is approximated by $\widetilde{U}_i$, the corresponding approximated circuit is given by $\widetilde{U}=\widetilde{U}_1 \widetilde{U}_2 \ldots \widetilde{U}_m$.
By the triangle inequality and the unitary invariance of the spectral norm, the accumulated approximation error satisfies
\begin{equation}\label{eq47}
\|U-\widetilde{U}\| \leq \sum_{i=1}^{m}\|U_i-\widetilde{U}_i\|.
\end{equation}
A similar error accumulation bound also holds for the trace distance due to the triangle inequality and the unitary invariance of the trace norm.

\section{Quantum Circuits for Isometries and Quantum Channels on Qudits}\label{sec5}

\subsection{Quantum Circuits of Isometries}\label{sec5.1}
Throughout this subsection, we assume that $n \leq m$.
We now show how to construct a quantum circuit for an arbitrary isometry
$V\in \mathrm{U}(\mathcal{H}_{d}^{\otimes n},\mathcal{H}_{d}^{\otimes m})$
using CINC gates and single-qudit gates.
Our construction builds upon the state preparation scheme proposed in \cite{Bullock2005} and the following lemma~\ref{Lem5.2}.

First, according to Ref.~\cite{Bullock2005}, an $n$-qudit state preparation circuit that maps the initial state $|0\rangle^{\otimes n}$ to an arbitrary $n$-qudit state requires at most $\frac{d^n-1}{d-1}-n$ single-controlled single-qudit gates.
By Lemmas~\ref{Lem3.1} and \ref{Lem3.2}, each single-controlled single-qudit gate can be decomposed into two CINC gates and, when $d$ is prime, into $d$ SUM gates.
This leads to the following theorem.

\begin{theorem}\label{The5.1}
For any given $n$-qudit state $|\psi\rangle$, there exists a quantum circuit that maps the initial state $|0\rangle^{\otimes n}$ to $|\psi\rangle$ using at most $2(\frac{d^n-1}{d-1}-n)$ CINC gates.
When $d$ is a prime, such a circuit requires at most $d(\frac{d^n-1}{d-1}-n)$ SUM gates.
\end{theorem}

\begin{lemma}\label{Lem5.2}
For any isometry $V\in \mathrm{U}(\mathcal{H}_{d}^{\otimes n},\mathcal{H}_{d}^{\otimes m})$, there exists a unitary $U\in \mathrm{U}(\mathcal{H}_{d}^{\otimes m})$ that satisfies 
\begin{equation}\label{eq49}
V|\psi\rangle=U(|0\rangle^{\otimes (m-n)} \otimes |\psi\rangle),
\end{equation}
for every $|\psi\rangle \in \mathcal{S}(\mathcal{H}_{d}^{\otimes n})$, and $U$ has at least $d^m-d^n$ eigenvalues with value $1$.
\end{lemma}

\begin{proof}
Since $V\in \mathrm{U}(\mathcal{H}_{d}^{\otimes n},\mathcal{H}_{d}^{\otimes m})$, $V$ corresponds to a $d^m \times d^n$ matrix with orthonormal columns in the computational basis.
From Lemma 3.2 of Ref.~\cite{Knill1995}, there exists a $d^m \times (d^m-d^n)$ matrix $W$ with orthonormal columns such that $(V,W)$ is a unitary matrix that has at least $d^m-d^n$ eigenvalues equal to 1.
Taking $U=(V,W)$, it holds that 
\begin{equation}\label{eq50}
U(|0\rangle^{\otimes (m-n)} \otimes |\psi\rangle)=(V,W) 
\left(
  \begin{array}{c}
    |\psi\rangle  \\
    \mathbf{0}    \\
  \end{array}
\right)
=V|\psi\rangle,
\end{equation}
for every $|\psi\rangle \in \mathcal{S}(\mathcal{H}_{d}^{\otimes n})$.
\end{proof}

Lemma~\ref{Lem5.2} shows that a quantum circuit for $V\in \mathrm{U}(\mathcal{H}_{d}^{\otimes n},\mathcal{H}_{d}^{\otimes m})$ can be obtained by constructing an $m$-qudit circuit for $U\in \mathrm{U}(\mathcal{H}_{d}^{\otimes m})$, where the first $(m-n)$ input qudits are all initialized to $|0\rangle$, as shown in Fig.~\ref{Fig.13}.
Combining Lemma~\ref{Lem5.2} with Theorem~\ref{The4.7}, one obtains the following theorem, the proof of which is similar to that of Theorem 1 in \cite{isometries}.

\begin{figure}[htbp]
  \centering
  \includegraphics[width=6.5cm]{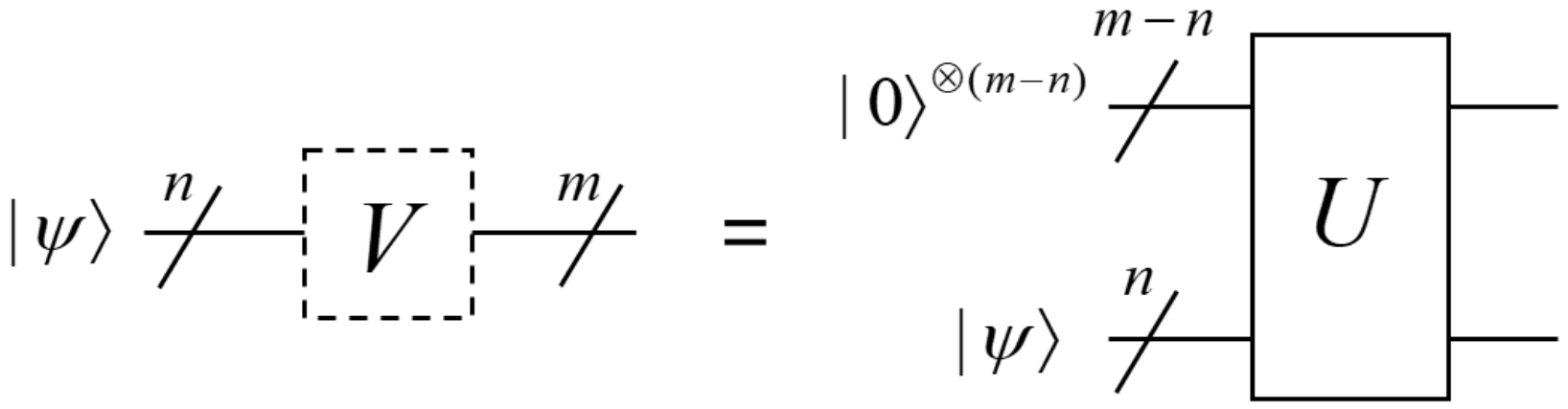}
  \caption{Quantum circuit model for an isometry. We use a dashed box to denote an isometry to distinguish it from a unitary.}\label{Fig.13}
\end{figure}

\begin{theorem}\label{The5.3}
For every isometry $V\in \mathrm{U}(\mathcal{H}_{d}^{\otimes n},\mathcal{H}_{d}^{\otimes m})$, there exists a quantum circuit for implementing $V$ using $m-n$ ancilla qudits initialized to $|0\rangle^{\otimes(m-n)}$, such that the number of CINC gates in the circuit is at most 
\begin{equation}\label{eq51}
N(d,n,m)=4d^n\Big(\frac{d^m-1}{d-1}-m\Big)+d^nN_{\mathrm{U}}(d,m),
\end{equation}
where $N_{\mathrm{U}}(d,m)$ is given in Theorem~\ref{The4.7}.
\end{theorem}

\begin{proof}
It suffices to bound the number of CINC gates required for the synthesis of $U\in \mathrm{U}(\mathcal{H}_{d}^{\otimes m})$, where $U$ is obtained from Lemma~\ref{Lem5.2}.
Let
\begin{equation}\label{eq52}
U=\sum_{a=0}^{d^n-1} e^{\textrm{i}\theta_a}|\psi_a\rangle \langle\psi_a| + 
\sum_{b=d^n}^{d^m-1} |\psi_b\rangle \langle\psi_b| 
\end{equation}
be a spectral decomposition of $U$.
For simplicity, we write $|a\rangle=|x_1 \dots x_m\rangle$, where $x_i \in [\underline{d}]$ and $a=\sum_{i=1}^{m} x_i d^{m-i}$.
For every $a \in \{0,\dots,d^n-1\}$, let $W_a \in \mathrm{U}(\mathcal{H}_{d}^{\otimes m})$ be the unitary operator such that $W_a|a\rangle=|\psi_a\rangle$.
Since all projection operators $|\psi_k\rangle \langle\psi_k|$ commute, it holds that
\begin{equation}\label{eq53}
\begin{split}
U&=\exp\Big(\sum_{a=0}^{d^n-1} \mathrm{i}\theta_a|\psi_a\rangle \langle\psi_a| + \sum_{b=d^n}^{d^m-1} \mathrm{i} \cdot 0 |\psi_b\rangle \langle\psi_b| \Big) \\
&=\prod_{a=0}^{d^n-1} \exp(\mathrm{i}\theta_a|\psi_a\rangle \langle\psi_a|)\\
&=\prod_{a=0}^{d^n-1} W_a \exp(\mathrm{i}\theta_a|a\rangle \langle a|) W_a^{\dag},
\end{split}
\end{equation}
where the second equality follows from the fact that $\exp(\mathbf{0}) = I$ for the terms where $b \geq d^n$.

According to Theorem~\ref{The5.1}, each $W_a$ (or $W_a^\dagger$) can be synthesized using $2(\frac{d^m-1}{d-1}-m)$ CINC gates.
Moreover, $\exp(\textrm{i}\theta_a|a\rangle \langle a|)$ is an $(m-1)$-controlled single-qudit gate, whose CINC gate count is bounded by $N_{\mathrm{U}}(d,m)$ as established in Theorem~\ref{The4.7}.
Summing the CINC counts, we arrive at the upper bound in \eqref{eq51}. 
\end{proof}

Similar to Corollary~\ref{Cor4.8}, for a prime $d$, combining Theorem~\ref{The5.3} with Lemma~\ref{Lem3.2} yields an upper bound on the number of SUM gates required to synthesize an isometry $V\in \mathrm{U}(\mathcal{H}_{d}^{\otimes n},\mathcal{H}_{d}^{\otimes m})$.

\begin{corollary}\label{Cor5.4}
Let $d$ be a prime number.
For every isometry $V\in \mathrm{U}(\mathcal{H}_{d}^{\otimes n},\mathcal{H}_{d}^{\otimes m})$, there exists a quantum circuit for implementing $V$ using $m-n$ ancilla qudits initialized to $|0\rangle^{\otimes(m-n)}$, such that the number of SUM gates in the circuit is at most $d N(d,n,m)$.
\end{corollary}

Assume that an $n$-qudit circuit consists of SUM (or CINC) gates and single-qudit gates, where the positions of all gates are fixed but the single-qudit gate at each position can be varied arbitrarily.
If such a quantum circuit can implement any unitary operation on the $n$-qudit system by varying the single-qudit gates, we refer to it as a \emph{universal $n$-qudit circuit}.
Corollary~\ref{Cor5.4} (resp. Theorem~\ref{The5.3}) gives an upper bound on the number of SUM (resp. CINC) gates required for a universal $n$-qudit circuit when $d$ is prime, while the next theorem provides a lower bound.

\begin{theorem}\label{The5.5}
For any universal $n$-qudit circuit as defined above, the number of SUM (or CINC) gates in the circuit is at least $\lceil\frac{1}{2d(d-1)}[d^{2n}-n(d^2-1)-1]\rceil$.
\end{theorem}

The proof of Theorem~\ref{The5.5} employs the real-parameter counting method used in \cite{Bullock2005,isometries,Shende2004}.
Specifically, a lower bound on the number of SUM gates is given by the ratio of the total free real parameters of an $n$-qudit unitary operation to the number of free real parameters that one SUM gate can introduce at most.
The detailed proof of Theorem~\ref{The5.5} is provided in Appendix~\ref{Appendix.A}.

\subsection{Quantum Circuits of Quantum Channels}\label{sec5.2}

Suppose that $\mathcal{N}$ is a quantum channel from $n$-qudit to $m$-qudit with Choi rank $K$ ($1\leq K \leq d^{n+m}$).
Note that in this subsection, $n$ is not necessarily less than $m$.
Let $l=\lceil \log_dK \rceil$. 
By the Stinespring representation of $\mathcal{N}$, there exists an isometry $V\in \mathrm{U}(\mathcal{H}_{d}^{\otimes n},\mathcal{H}_{d}^{\otimes (l+m)})$ such that for every $\rho\in \mathcal{D}(\mathcal{H}_{d}^{\otimes n})$,
\begin{equation}\label{eq54}
\mathcal{N}(\rho)=\textrm{Tr}_{\mathcal{H}_{d}^{\otimes l}}(V \rho V^{\dag}),
\end{equation}
where $\textrm{Tr}_{\mathcal{H}_{d}^{\otimes l}}$ denotes the partial trace over the first $l$ qudits.
The partial trace over the first $l$ qudits can be implemented by measuring those qudits in the computational basis and discarding the measurement outcomes. 
Hence, a circuit model for $\mathcal{N}$ is obtained by taking the circuit for $V$ from Theorem~\ref{The5.3} (or Corollary~\ref{Cor5.4}) and then performing such a measurement on the first $l$ output qudits.
This leads directly to the following theorem.

\begin{theorem}\label{The5.6}
Let $\mathcal{N}$ be a quantum channel from $n$-qudit to $m$-qudit with Choi rank $K$.
Using $\lceil \log_dK \rceil+m-n$ ancilla qudits each initialized to $|0\rangle$, the number of CINC gates required for a quantum circuit model of $\mathcal{N}$ is at most $N(d,n,\lceil \log_dK \rceil+m)$.
Moreover, if $d$ is prime, the number of SUM gates required is at most $dN(d,n,\lceil \log_dK \rceil+m)$.
Here $N(\cdot,\cdot,\cdot)$ is given by Eq.~\eqref{eq51}.
\end{theorem}

From the expression of $N(d,n,\lceil \log_dK \rceil+m)$, the quantum circuit for $\mathcal{N}$ requires $O(Kd^{n+m})$ CINC gates (or SUM gates when $d$ is prime).
In the worst case where $K=d^{n+m}$, this becomes $O(d^{2(n+m)})$, which is prohibitively large. 
However, $\lceil \log_dK \rceil$ qudits are measured at the end of the circuit.
If a classical control is allowed after each measurement (i.e., the measurement outcome can be used to conditionally determine which unitary is applied to the target qudits), the number of CINC gates can be reduced.
This type of circuit is referred to as \emph{MeasuredQCM} (measured quantum circuit model) in Ref.~\cite{channels}.

\begin{theorem}\label{The5.7}
Let $\mathcal{N}$ be a quantum channel from $n$-qudit to $m$-qudit with Choi rank $K>d$.
Using $\lceil \log_dK \rceil+m-n$ ancilla qudits each initialized to $|0\rangle$, the number of CINC gates required for a MeasuredQCM of $\mathcal{N}$ is at most 
\begin{equation}\label{eq55}
\begin{cases}
\lceil \log_dK \rceil N(d,n,n+1)+N(d,n,m),    & \text{if } n < m,   \\[1mm]
(\lceil \log_dK \rceil+m-n)N(d,n,n+1),        & \text{if } n \geq m.   
\end{cases}
\end{equation}
Moreover, if $d$ is prime, the number of SUM gates required is at most the corresponding upper bounds in \eqref{eq55} multiplied by $d$.
\end{theorem}

The proof of Theorem~\ref{The5.7} is a direct generalization of the qubit-based MeasuredQCM in Ref.~\cite{channels}, except that our synthesis of isometries (Theorem~\ref{The5.3}) in the qudit case replaces the qubit one.
For completeness, the detailed proof is given in Appendix~\ref{Appendix.B}.

\section{Conclusion}\label{sec6}

In this paper, we have presented efficient quantum circuit synthesis schemes for multi-controlled single-qudit gates, isometries, and quantum channels in general $d$-level systems.
By introducing the multi-controlled pseudo-increment gate $\widetilde{X}_d$, we reduce the CINC and GCX gate counts for synthesizing general $(n-1)$-controlled unitaries to $O(n^2)$ and further compress the complexity to $O(n)$ for special unitaries.
This result improves upon the previous best bound of $O(n^{2+\log_2 d})$ CINC gates \cite{Brennen2006}.
Since each CINC gate can be decomposed into $d-1$ GCX gates, our construction reduces the GCX count from $O(n^3)$ to $O(n^2)$ \cite{Di2013}.
This achieves the same asymptotic bound as in the qubit case.
We also provide an approximate linear-cost construction and an exact linear-cost construction using one ancilla qudit.
Furthermore, when $d$ is prime, we establish a synthesis scheme showing that all CINC-based operations can be compiled into SUM-gate circuits while preserving their asymptotic complexity.

Applying these techniques, we construct circuit architectures for $n$-to-$m$ qudit isometries and quantum channels, and employ the MeasuredQCM framework with classical controls to reduce the gate complexity of quantum channels.
Finally, using a real-parameter counting method, we prove a theoretical lower bound of $\lceil\frac{1}{2d(d-1)}(d^{2n}-nd^2+n-1)\rceil$ on the number of SUM (or CINC) gates required for universal $n$-qudit circuits.
These results provide practical, scalable design guidelines for high-dimensional quantum hardware and fault-tolerant architectures.
Future research directions include designing $n$-qudit circuits using SUM and single-qudit gates for non-prime $d$ and optimizing the circuit depth.

 
\section*{Acknowledgements} \par

This work is supported by the National Natural Science Foundation of China (Grant No. 62571166).

\section*{Appendix}

\appendix

\section{Proof of Theorem \ref{The5.5}}\label{Appendix.A}

Before presenting the proof, we give a decomposition of a unitary matrix, which is essentially a QR decomposition; a similar decomposition also appears in \cite{MZI1994,MZI2016}.
First given $1\leq i<j\leq n$ and $\phi,\theta \in \mathbb{R}$, define a $d \times d$ unitary matrix $G_{ij}(\phi,\theta)$ as follows: the entries with indices $i,j$ form the block 
\begin{equation}\label{eqA1}
\begin{split}
&\left(
  \begin{array}{cc}
    (G_{ij}(\phi,\theta))_{ii}  & (G_{ij}(\phi,\theta))_{ij}   \\[1mm]
    (G_{ij}(\phi,\theta))_{ji}  & (G_{ij}(\phi,\theta))_{jj}   \\
  \end{array}
\right)\\
&=\left(
  \begin{array}{cc}
    e^{\mathrm{i}\phi}\cos\theta &  \sin\theta     \\[1mm]
    -\sin\theta  & e^{-\mathrm{i}\phi}\cos\theta   \\
  \end{array}
\right),
\end{split}
\end{equation}
while $(G_{ij}(\phi,\theta))_{kk}=1$ for all $k\notin\{i,j\}$, and all remaining entries are zero.
When the specific values of $\phi$ and $\theta$ are irrelevant or clear from context, we simply write $G_{ij}$ for brevity.

For any $d \times d$ unitary matrix $U$, we multiply $U$ on the left by a suitable $G_{d-1,d}(\phi,\theta)$ such that the $(d-1,d)$ entry of $U$ becomes zero.
Then we multiply $G_{d-1,d}U$ on the left by $G_{d-2,d}$ such that it zeroes the $(d-2,d)$ entry and leaves the previously zeroed $(d-1,d)$ entry unchanged.
We continue in this fashion until all entries of the last column except the diagonal entry are zero.
Thus we obtain a sequence of unitary matrices $G_{1,d}, \ldots, G_{d-2,d},  G_{d-1,d}$ such that
\begin{equation}\label{eqA2}
G_{1,d} \ldots G_{d-1,d} U = \big(\prod _{i=1}^{d-1}G_{i,d}\big) U=
\left(
  \begin{array}{cc}
    V  & \mathbf{0}  \\
    \mathbf{0}  & e^{\mathrm{i}\theta}  \\
  \end{array}
\right),
\end{equation}
where $V$ is a $(d-1) \times (d-1)$ unitary matrix.
Similarly, we can multiply \eqref{eqA2} on the left by $\prod_{i=1}^{d-2}G_{i,d-1}$ such that the entries in the $(d-1)$-th column above the diagonal become zero, and the last row and column will not be affected.
By repeating this process for columns $d-2,\ldots,2$, we obtain  
\begin{equation}\label{eqA3}
\big(\prod _{j=2}^{d} \prod_{i=1}^{j-1}G_{ij}\big)  U = D,
\end{equation}
where $D$ is a diagonal unitary matrix.
Thus, we obtain the decomposition of $U$ as follows
\begin{equation}\label{eqA4}
U = G D, \;\; G=\Big(\prod _{j=2}^{d} \prod_{i=1}^{j-1}G_{ij}\Big)^{\dag}.
\end{equation}
Since $(G_{ij}(\phi,\theta))^{\dag}=G_{ij}(-\phi,-\theta)$, Eq.~\eqref{eqA4} implies that any $d \times d$ unitary matrix $U$ can always be decomposed into a product of $d(d-1)$ two-parameter matrices $G_{ij}$ and a diagonal unitary matrix.

\begin{proof}[Proof of Theorem \ref{The5.5}]
By Eq.~\eqref{eq2}, it holds that
\begin{equation}\label{eqA5}
\begin{split}
I_d \otimes F_d \cdot \mathrm{SUM} \cdot I_d \otimes F_d^{\dag}
&=\sum_{a\in [\underline{d}]} |a\rangle\langle a| \otimes (F_{d} X_{d}^{a} F^{\dag}_{d}) \\
&=\sum_{a\in [\underline{d}]} |a\rangle\langle a| \otimes Z_{d}^{a} =: \Lambda (Z_{d}).
\end{split}
\end{equation}
This implies that a lower bound on the number of $\Lambda (Z_{d})$ required for a universal quantum circuit must also be a lower bound on the number of SUM, and vice versa.
Therefore, it suffices to discuss the lower bound on the number of $\Lambda (Z_{d})$ required for a universal $n$-qudit circuit.

Assume that an $n$-qudit circuit $\mathcal{T}$ consists of unspecified single-qudit gates and $k$ $\Lambda (Z_{d})$ gates.
Since the two adjacent single-qudit gates can be combined into one, one may assume without loss of generality that $\mathcal{T}$ contains $n+2k$ single-qudit gates.
Next we show that $\mathcal{T}$ can be replaced with an $n$-qudit circuit $\mathcal{T}'$ that consists of $k$ $\Lambda (Z_{d})$ gates, $d(d-1)(n+2k)$ two-parameter gates $G_{ij}$, $n$ diagonal special unitary gates (each depending on $d-1$ real parameters), and a global phase gate $e^{\mathrm{i}\alpha}I_d$.

By Eq.~\eqref{eqA4}, every single-qudit gate can be decomposed into the form $GD$. It therefore suffices to prove that $n+2k$ diagonal gates can be reduced to $n$ diagonal special unitary gates and a global phase gate.
First, for any $U_1,U_2\in \mathrm{U}(\mathcal{H}_{d})$, it holds that 
\begin{equation}\label{eqA6}
\begin{split}
U_1 \otimes U_2 \cdot \Lambda (Z_{d}) 
&= G_1D_1 \otimes G_2D_2 \cdot \Lambda (Z_{d}) \\
&= G_1 \otimes G_2 \cdot \Lambda (Z_{d}) \cdot D_1 \otimes D_2,
\end{split}
\end{equation}
where the first equality follows from Eq.~\eqref{eqA4} and the second equality follows from the fact that $\Lambda (Z_{d})$ and $D_1 \otimes D_2$ are diagonal.
Equation~\eqref{eqA6} shows that, for each $\Lambda (Z_{d})$ in $\mathcal{T}$, the single-qudit gate to its right can be decomposed as $GD$, and then the diagonal part $D$ can be moved to the left of $\Lambda (Z_{d})$.
Thus, this $D$ can be combined with the single-qudit gate on the left of $\Lambda (Z_{d})$.
By iterating this process, there will be $2k$ single-qudit gates in $\mathcal{T}$ that can be replaced with gates of the form $G$.
In addition, the $n$ single-qudit gates at the far left of the circuit $\mathcal{T}$ can be decomposed as $GD$, and for any $n$ diagonal unitary gates $D_1,\ldots,D_n$, there exists a real parameter $\alpha$ and $n$ diagonal special unitary gates $D'_1,\ldots,D'_n$ such that 
\begin{equation}\label{eqA7}
D_1 \otimes \ldots \otimes D_n = e^{\mathrm{i}\alpha}D'_1 \otimes \ldots \otimes D'_n.
\end{equation}
Therefore, $\mathcal{T}$ can indeed be replaced by $\mathcal{T}'$.

The unitary gates that can be simulated by the circuit $\mathcal{T}'$ depend on
\begin{equation}\label{eqA8}
m=2d(d-1)(n+2k)+(d-1)n+1
\end{equation}
independent real parameters.
Thus, one can define a map $F:\mathbb{R}^{m} \rightarrow \mathbb{C}^{d^n \times d^n}$ (where $\mathbb{C}^{d^n \times d^n}$ denotes the set of all $d^n \times d^n$ complex matrices) by sending every $\boldsymbol{\theta} \in \mathbb{R}^{m}$ to the unitary matrix that corresponds to the unitary gate simulated by $\mathcal{T}'$ with parameters $\boldsymbol{\theta}$ (in the computational basis).
Since $e^{\mathrm{i}\theta}$, $\cos\theta$ and $\sin\theta$ are smooth functions of $\theta\in \mathbb{R}$, $F$ is a smooth map from $\mathbb{R}^{m}$ to the smooth manifold $\mathbb{C}^{d^n \times d^n}$.
Let $\mathrm{U}(d^n)$ denote the set of all $d^n \times d^n$ unitary matrices.
Since $\mathrm{U}(d^n)$ is an embedded submanifold of $\mathbb{C}^{d^n \times d^n}$ with dimension $d^{2n}$ over $\mathbb{R}$, $F$ is also a smooth map from $\mathbb{R}^{m}$ to the smooth manifold $\mathrm{U}(d^n)$ by Corollary 5.30 in \cite{Lee}.

Suppose $\mathcal{T}$ is a universal $n$-qudit circuit. This implies that $F(\mathbb{R}^{m})=\mathrm{U}(d^n)$, where $F(\mathbb{R}^{m})$ denotes the image of $F$.
By Sard’s Theorem (see \cite[Corollary 6.11]{Lee}), it must hold that $m \geq d^{2n}$.
Otherwise, $F(\mathbb{R}^{m})$ has measure zero in $\mathrm{U}(d^n)$, which contradicts $F(\mathbb{R}^{m})=\mathrm{U}(d^n)$.
Therefore, for a universal $n$-qudit circuit containing $k$ $\Lambda (Z_{d})$ gates, we have 
\begin{equation}\label{eqA8}
k \geq \frac{1}{2d(d-1)}[d^{2n}-(d^2-1)n-1].
\end{equation}

A similar lower bound holds for universal $n$-qudit circuits composed of CINC gates. 
Notice that the CINC gate $\mathrm{C}^{2}_{1,2}(X_d)$ is equivalent to $\mathrm{C}^{2}_{1,2}(Z_d)$ up to single-qudit Fourier transform gates $F_d$, and $\mathrm{C}^{2}_{1,2}(Z_d)$ also satisfies the commutativity relation given in Eq.~\eqref{eqA6}.
Therefore, the counting of independent parameters remains identical, which yields the same lower bound.
\end{proof}

\section{Proof of Theorem \ref{The5.7}}\label{Appendix.B}

From Eq.~\eqref{eq54}, we first give a quantum circuit of the isometry $V\in \mathrm{U}(\mathcal{H}_{d}^{\otimes n},\mathcal{H}_{d}^{\otimes (l+m)})$, where $l=\lceil \log_dK \rceil$. 
In the computational basis, one can write $V$ as 
\begin{equation}\label{eqB1}
V=\left(
  \begin{array}{c}
    A_{1}  \\
    A_{2}  \\
    \vdots  \\
    A_{d}  \\
  \end{array}
\right),
\end{equation}
where $A_{i}$ is a $d^{l+m-1} \times d^{n}$ matrix and $\sum_{i=1}^{d}A_{i}^{\dag}A_{i}=I_{d^{n}}$.
Applying the QR-decomposition to each $A_{i}$, one may write $A_{i}=Q_{i}R_{i}$ for $Q_{i}$ being a $d^{l+m-1} \times d^{l+m-1}$ unitary matrix and $R_{i}$ being a $d^{l+m-1} \times d^{n}$ upper triangular matrix.
Let 
\begin{equation}\label{eqB2}
Q=\left(
  \begin{array}{cccc}
    Q_{1} &  &  &  \\
     & Q_{2} &  &  \\
     &  & \ddots &  \\
     &  &  & Q_{d} \\
  \end{array}
\right),\;\;
R=\left(
  \begin{array}{c}
    R_{1}  \\
    R_{2}  \\
    \vdots  \\
    R_{d}  \\
  \end{array}
\right),
\end{equation}
where all unspecified entries in the block-diagonal matrix are zeros.
It follows that $V=QR$.

As each $R_{i}$ is a $d^{l+m-1} \times d^{n}$ upper triangular matrix, one may write
\begin{equation}\label{eqB3}
R_{i}=\left(
  \begin{array}{c}
    T_{i1}     \\
    \mathbf{0} \\
    \vdots \\
    \mathbf{0} \\
  \end{array}
\right),
\end{equation}
where $T_{i1}$ is a $d^{n} \times d^{n}$ upper triangular matrix.
It follows from $R\in \mathrm{U}(\mathcal{H}_{d}^{\otimes n},\mathcal{H}_{d}^{\otimes (l+m)})$ that 
\begin{equation}\label{eqB4}
\left(
  \begin{array}{c}
    T_{11}     \\
    T_{21} \\
    \vdots \\
    T_{d1} \\
  \end{array}
\right)\in \mathrm{U}(\mathcal{H}_{d}^{\otimes n},\mathcal{H}_{d}^{\otimes (n+1)}),
\end{equation}
By Lemma \ref{Lem5.2}, there exists a unitary operator
\begin{equation}\label{eqB5}
T=\left(
  \begin{array}{cccc}
    T_{11} & T_{12} & \cdots & T_{1d}\\
    T_{21} & T_{22} & \cdots & T_{2d}\\
    \vdots & \vdots & \ddots & \vdots\\
    T_{d1} & T_{d2} & \cdots & T_{dd}\\
  \end{array}
\right)\in \mathrm{U}(\mathcal{H}_{d}^{\otimes (n+1)}),
\end{equation}
such that $T$ has at least $d^{n+1}-d^n$ eigenvalues with value $1$.
Let
\begin{equation}\label{eqB6}
\widetilde{T}_{ij}=\left(
  \begin{array}{cccc}
    T_{ij} & \mathbf{0}  \\
    \mathbf{0} & I_{d^{l+m-1}-d^n}\\
  \end{array}
\right),
\end{equation}
so that
\begin{equation}\label{eqB7}
\widetilde{R}=\left(
  \begin{array}{cccc}
    \widetilde{T}_{11} & \widetilde{T}_{12} & \cdots & \widetilde{T}_{1d}\\
    \widetilde{T}_{21} & \widetilde{T}_{22} & \cdots & \widetilde{T}_{2d}\\
    \vdots & \vdots & \ddots & \vdots\\
    \widetilde{T}_{d1} & \widetilde{T}_{d2} & \cdots & \widetilde{T}_{dd}\\
  \end{array}
\right)\in \mathrm{U}(\mathcal{H}_{d}^{\otimes (l+m)}).
\end{equation}
Let $\widetilde{V}=Q\widetilde{R}$. It is evident that
\begin{equation}\label{eqB8}
V |\psi\rangle=\widetilde{V} (|0\rangle^{\otimes (l+m-n)} \otimes |\psi\rangle),
\end{equation}
for every $|\psi\rangle \in \mathcal{S}(\mathcal{H}_{d}^{\otimes n})$.
Thus, a quantum circuit for $V$ can be obtained by constructing a circuit for $\widetilde{V}$ and preparing the first $(l+m-n)$ input qudits as $|0\rangle^{\otimes (l+m-n)}$.

Observe that $\widetilde{R}$ is a multi-controlled $T$ gate, where the control qudits are the second to the $(l+m-n)$-th qudits (with control state $|0\rangle^{\otimes (l+m-n-1)}$) and the target qudits are the first and the last $n$ qudits.
Hence, when the input state $|0\rangle^{\otimes (l+m-n)} \otimes |\psi\rangle$ passes through $\widetilde{R}$, it essentially applies $T$ to the first and the last $n$ qudits, while the second to the $(l+m-n)$-th qudits remain in the state $|0\rangle$.
As for $Q$, from the definition of $Q$, it acts as follows: if the first qudit is in state $|i-1\rangle$ ($i\in \{1,\ldots,d\}$), then $Q_{i}$ is applied to the remaining qudits.
Since the second to the $(l+m-n)$-th qudits are in the state $|0\rangle$ after applying $\widetilde{R}$, $Q_{i}$ can be regarded as an isometry $Q'_{i}\in\mathrm{U}(\mathcal{H}_{d}^{\otimes n},\mathcal{H}_{d}^{\otimes (l+m-1)})$, i.e., the action of $Q_{i}$ can be replaced by $Q'_{i}$.
Then one can apply the same procedure as one did for $V$ to $Q'_{i}$.
This procedure is repeated $l$ (resp. $l+m-n-1$) times when $n<m$ (resp. $n \geq m$).
After these steps, a $(l+m)$-qudit circuit for $V$ is obtained, as shown in Fig.~\ref{14a} for $n<m$ (for $n \geq m$, replace $m-n$ by 1 and $l$ by $l+m-n-1$ in the figure).

Finally, we need to measure the first $l$ qudits in the computational basis and forget the measurement outcomes.
Because classical control after measurement is allowed, the measurement commutes with the controlled gates.
Therefore, we obtain a quantum circuit for the channel $\mathcal{N}$ as shown in Fig.~\ref{14b}.
By Theorem~\ref{The5.3}, the number of CINC gates required for this circuit is at most
\begin{equation}\label{eqB9}
\begin{cases}
l N(d,n,n+1)+N(d,n,m),                   & \text{if } n < m,   \\
(l+m-n)N(d,n,n+1),                       & \text{if } n \geq m.   
\end{cases}
\end{equation}

\begin{figure*} [htbp]
  \centering
  \subfigure[]{
  \includegraphics[width=6.5cm]{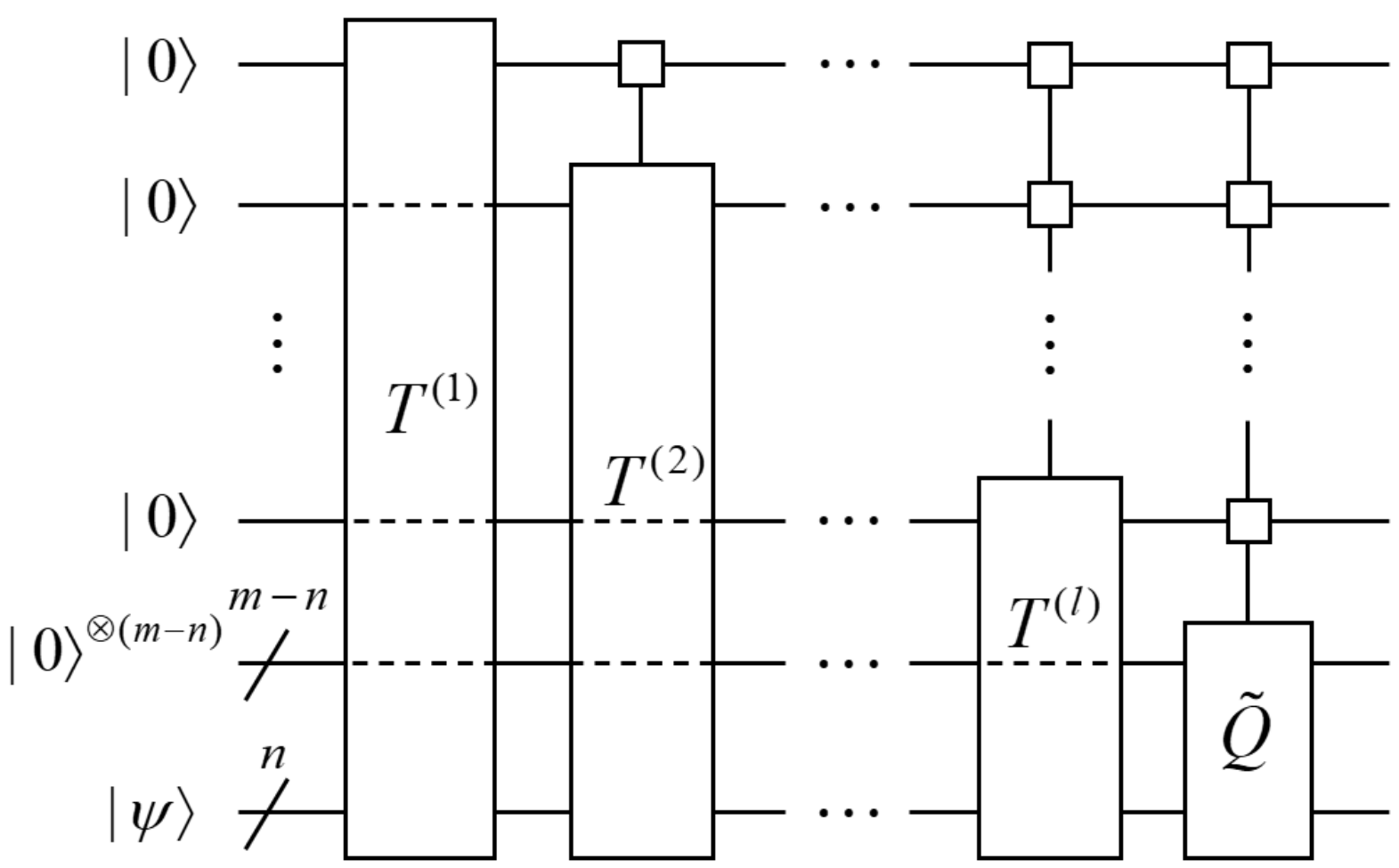}
  \label{14a}}\;
  \subfigure[]{
  \includegraphics[width=8cm]{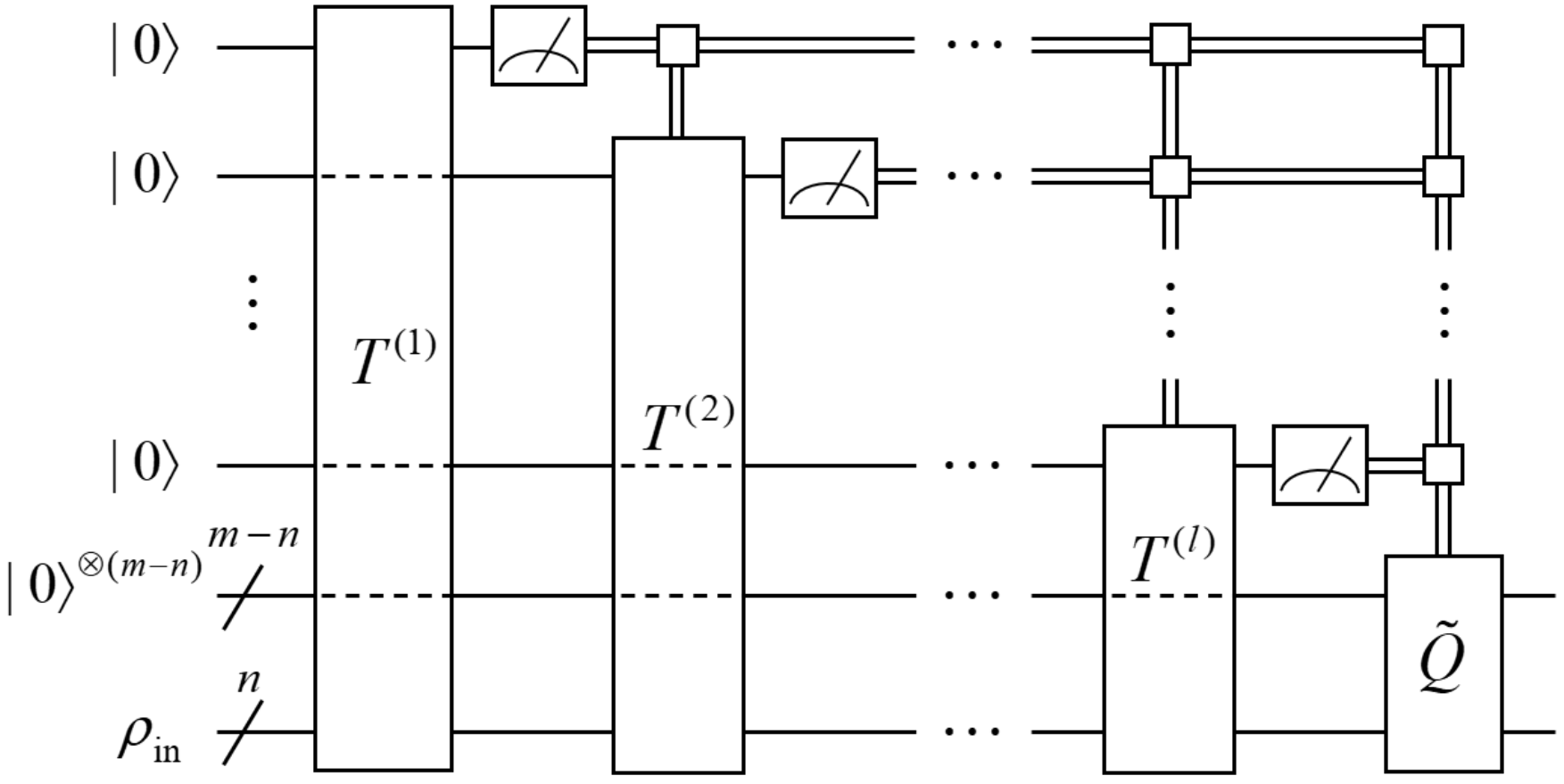}
  \label{14b}}
  \caption{(a) Quantum circuit for the isometry $V\in \mathrm{U}(\mathcal{H}_{d}^{\otimes n},\mathcal{H}_{d}^{\otimes (l+m)})$ with $n<m$. The unfilled square on the wire denotes that the gate is controlled by the qudit in such a way that the applied unitary on the target may depend on the control state. The dashed line indicates that the corresponding qudit is not affected by the gate (i.e., it acts as the identity). (b) Measured quantum circuit model for the channel $\mathcal{N}$ from $n$-qudit to $m$-qudit with $n<m$. The double lines denote classical wires carrying the results of measurements, which may be used to control subsequent gates. }
  \label{Fig.14}
\end{figure*}


\end{document}